\documentclass{article}
\usepackage{etex}
\usepackage{longtable}
\usepackage{rotating}
\expandafter\def\csname ver@subfig.sty\endcsname{}

\usepackage{svg}

\usepackage{array}              % Extensions for the tabular formatting
\usepackage{booktabs}           % Top-, mid- and bottomrule

\usepackage{soul}
\soulregister\cite7
\soulregister\ref7
\soulregister\pageref7
\soulregister{\em}{0}
\soulregister{\mdash}{0}
\soulregister{\st}{7}
\soulregister{\bf}{0}
\soulregister{\fig}{7}
\soulregister{\ie}{0}
\soulregister{\etc}{0}
\soulregister{\eg}{0}
\soulregister{\wrt}{0}
\soulregister{\resp}{0}
\soulregister{\cf}{0}

\usepackage[colorinlistoftodos,bordercolor=white]{todonotes}

\usepackage{hyperref,url}
\usepackage{fancyhdr}
\usepackage{graphicx}
\usepackage{fullpage}

\usepackage{stfloats}
\usepackage{varwidth}
\usepackage[ruled,vlined]{algorithm2e}

\usepackage{graphicx,color}
\usepackage{mathtools}
\usepackage{amscd}
\usepackage{amsthm}
\usepackage{amsmath}
\usepackage{latexsym,amssymb}
\usepackage{wrapfig}
 \usepackage{multirow}
 \usepackage{pifont}
\usepackage{rotating}

 \usepackage{listings, color} % for code listings

\definecolor{dkgreen}{rgb}{0,0.6,0}
\definecolor{forestgreen}{RGB}{0,100,50}
\definecolor{redd}{RGB}{76,0,153}

\lstdefinestyle{customxml}
{language=XML,
keywordstyle=\color{green},
basicstyle=\ttfamily\scriptsize,
morekeywords={id, specType},
mathescape=true,
escapeinside={/*@}{@*/},
tagstyle=\color{forestgreen},
keywordstyle=\color{redd},
stringstyle=\ttfamily\color{blue},
frame=single
}

\usepackage{url} 
\newcounter{tempctr}
\newcommand{\breakenumistart}{%
  \setcounter{tempctr}{\value{enumi}}%
  \end{enumerate}%
}
\newcommand{\breakenumiend}{%
  \begin{enumerate}%
  \setcounter{enumi}{\value{tempctr}}%
}

\newtheorem{theorem}{Theorem}[section]
\newtheorem{proposition}[theorem]{Proposition}%[section]
\newtheorem{corollary}[theorem]{Corollary}%[section]
\newtheorem{definition}{Definition}%[section]
\newtheorem{remark}[theorem]{Remark}%[section]
\newtheorem{syntax}{Syntax}
\newtheorem{semantics}{Semantics}
\newtheorem{example}{Example}%[section]

\newcommand{\defn}[1]{Def.~\ref{defn:#1}}
\newcommand{\Defn}[1]{Definition~\ref{defn:#1}}

\newcommand{\fig}[1]{Fig.~\ref{fig:#1}}
\newcommand{\Fig}[1]{Figure~\ref{fig:#1}}

\newcommand{\ex}[1]{Ex.~\ref{ex:#1}}

\newcommand{\tabsize}{\fontsize{10}{10.5pt}\selectfont}
\newcommand{\resp}[1][\xspace]{resp.#1}

\newcommand{\mdash}{---}

\newcommand{\cf}[1][\ ]{cf.#1}
\newcommand{\ie}[1][\ ]{i.e.,#1}
\newcommand{\etc}[1][\ ]{etc.#1}
\newcommand{\eg}[1][\ ]{e.g.,#1}
\newcommand{\wrt}[1][\ ]{w.r.t.#1}

\newcommand{\bydef}[1]{\ensuremath{\stackrel{def}{#1}}}
\newcommand{\oftype}{\ensuremath{\!:\!}}

\newcommand{\non}[1]{\ensuremath{\overline{#1}}}

\newcommand{\goesto}[1][]{\ensuremath{\stackrel{#1}{\rightarrow}^{}}}

\newcommand{\true}{\ensuremath{\mathit{true}}}
\newcommand{\false}{\ensuremath{\mathit{false}}}

\newcommand{\require}{\ensuremath{\ \ \mathbf{Require}\ \ }}
\newcommand{\accept}{\ensuremath{\ \ \mathbf{Accept}\ \ }}

\newcommand{\cB}{\ensuremath{\mathcal{B}}}

\newcommand{\cT}{\ensuremath{\mathcal{T}}}

\newcommand{\compi}[1]{\ensuremath{\overline{#1}\:}}

\usepackage[titletoc]{appendix}

\graphicspath{{../images/}}

\newlength{\headheightold}

\newcommand{\coverpage}{
\pagestyle{fancy}
\pagenumbering{roman}

\fancyhead[R]{\includegraphics[height=1.8cm]{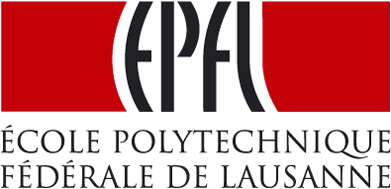}}
\setlength{\headheightold}{\headheight}
\addtolength{\textheight}{\headheight}
\setlength{\headheight}{2cm} 
\addtolength{\textheight}{-2cm}
\renewcommand{\headrulewidth}{0.4pt}
\renewcommand{\footrulewidth}{0pt}

\fancyhead[C]{\includegraphics[height=1.25cm]{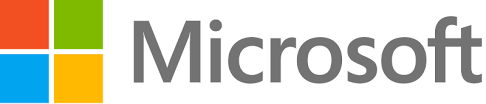}}
\setlength{\headheightold}{\headheight}
\addtolength{\textheight}{\headheight}
\setlength{\headheight}{2cm} 
\addtolength{\textheight}{-2cm}
\renewcommand{\headrulewidth}{0.4pt}
\renewcommand{\footrulewidth}{0pt}

\fancyhead[L]{\includegraphics[height=2.4cm]{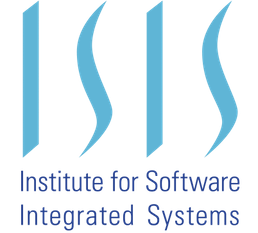}}
\setlength{\headheightold}{\headheight}
\addtolength{\textheight}{\headheight}
\setlength{\headheight}{2cm} 
\addtolength{\textheight}{-2cm}
\renewcommand{\headrulewidth}{0.4pt}
\renewcommand{\footrulewidth}{0pt}

\vspace*{20mm}
\begin{center}
  \begin{minipage}{0.8\textwidth}
    \centering
    \Huge \bf \mytitle

    \bigskip
    \hrule

    \vspace*{5mm}
    \large
    \myauthors

    \bigskip
    \today
  \end{minipage}
\end{center}

    \vspace*{25mm}
\paragraph{Abstract:}{\small \myabstract}
\vfill
\clearpage
\pagestyle{plain}
\setlength{\headheight}{\headheightold}
}
\newcommand{\mytitle}{Coordination of Dynamic Software Components with JavaBIP}
\newcommand{\myauthors}{Anastasia Mavridou, Valentin Rutz, Simon Bliudze}
\newcommand{\reportno}{******}
\newcommand{\myabstract}{JavaBIP allows the coordination of software components by clearly separating the functional and coordination aspects of the system behavior. JavaBIP implements the principles of the BIP component framework rooted in rigorous operational semantics. Recent work both on BIP and JavaBIP allows the coordination of static components defined prior to system deployment, \ie the architecture of the coordinated system is fixed in terms of its component instances.  Nevertheless, modern systems, often make use of components that can register and deregister dynamically during system execution. In this paper, we present an extension of JavaBIP than can handle this type of dynamicity. We use first-order interaction logic to define synchronization constraints based on component types. Additionally, we use directed graphs with coloring edges to model dependencies among components that determine the validity of an online system. We present the software architecture of
our implementation; provide and discuss performance evaluation results. }

\begin{document}

\coverpage

\thispagestyle{empty}
\vspace*{\fill}
\begin{verbatim}
@TechReport{DynamicJavaBIP,
Author = {Anastasia Mavridou and Valentin Rutz and Simon Bliudze},
Title = {Coordination of {D}ynamic {S}oftware {C}omponents with {J}ava{BIP}},
note     = {available at \url{https://arxiv.org/abs/1707.09716}},
Year = {2017},
Eprint = {arXiv:1707.09716}
}

\end{verbatim}
\clearpage
\pagenumbering{arabic}
\setcounter{page}{1}

\tableofcontents
\clearpage

\section{Introduction}
We have previously introduced JavaBIP~\cite{MiSE14p25,SPE:SPE2495} that allows coordinating software components exogenously, \ie without requiring access to component source code. JavaBIP relies on the following observations. Domain specific components have states (\eg idle, working) that are known to component users with domain expertise.  Furthermore, components always provide APIs that allow programs to invoke operations (\eg suspend or resume) in order to change their state, or to be notified when a component changes its state spontaneously.  Thus, component behavior can be easily represented by Finite State Machines (FSMs). 

JavaBIP brings the BIP principles into a more general software engineering context than that of embedded systems, in which code generation might not be desirable due to continuous code updates. Thus, to use JavaBIP, instead of generating Java code from the BIP modeling language, developers must provide\mdash for the relevant components\mdash the corresponding FSMs in the form of annotated Java classes. The FSMs describe the protocol that must be respected to access a shared resource or use a service provided by a component. FSM transitions are associated with calls to API functions, which force a component to take an action, or with event notifications that allow reacting to external events. 

For component coordination, JavaBIP provides two primitive mechanisms: 1)~multi-party synchronizations of component transitions and 2)~asynchronous event notifications. The latter embodies the reactive programming paradigm. In particular, JavaBIP extends the Actor model~\cite{actorsAgha}, since event notifications can be used to emulate asynchronous messages, while providing the synchronization of component transitions as a primitive mechanism gives developers a powerful and flexible tool to manage coordination. The synchronization of component transitions is managed by a runtime called JavaBIPEngine, which, for simplicity, we call ``engine'' in the rest of the paper. Notice that in a completely asynchronous system the engine is not needed.

JavaBIP clearly separates system-wide coordination policies from component behavior.
Synchronization constraints, defining the possible synchronizations among transitions of different components  \ie the set of possible component \emph{interactions},
are specified independently from the design of individual components in dedicated XML files.  %For coordination scenarios that require global state information, dedicated \emph{monitor} components can be added. This allows one to centralize all the information related to coordination in one single location, instead of distributing it across the components. 
This separation of functional and coordination aspects greatly reduces the burden of system complexity. Finally, integration with the BIP framework, through a JavaBIP to BIP code generation tool, allows the use of existing deadlock-detection and model checking tools~\cite{dfinder,esst4bip} ensuring the correctness of JavaBIP systems.

The previous implementation of JavaBIP~\cite{SPE:SPE2495} was static. To coordinate a system, the full set of components had be registered before starting the engine. No components could be added on-the-fly and, most importantly, if a failure occurred in a single component, the engine execution had to stop and the full set of constraints had to be computed anew. Notice that none of the current BIP implementations~\cite{bip,bip06,quilbeuf10-distr} allows to add or remove components on-the-fly, including DyBIP presented in~\cite{bozga2012modeling} that allows dynamically changing the set of interactions among a fixed set of components at runtime. This might be problematic, since modern systems, \eg large banking systems or modular smartphones, make use of components that can register and deregister during system execution.

To allow dynamicity in JavaBIP, we use first-order interaction logic to describe synchronization constraints on component types. As a result, a developer can write synchronization constraints without knowing the exact number of components in the system. Thus, component instances of known types, \ie types for which synchronization constraints exist, can register at runtime without any additional input from the developer. To optimize JavaBIP performance, we have introduced a notion of system validity: \emph{a system is valid if and only if its set of possible interactions is not empty}. The notion of validity allows to start and stop the engine automatically at runtime by just checking the status of the system. By stopping the engine if the system is invalid, we eliminate any processing time needed by the engine. To check system validity, we use directed graphs with edge coloring to model component synchronization dependencies. Notice that the introduced notion of validity is only relevant for the engine: in an invalid system components can still communicate asynchronously.

%The slightest modification of the program structure in terms of components and connections among them requires stopping the system, applying the required modifications and relaunching it. This means that it allows coordinating only software components that are statically defined prior to system deployment. This is extremely problematic, since modern systems, including large banking systems or systems designed for modular smartphones, make use of components that can register and deregister during the system execution.

We have extended the interface and implementation of the engine to register, deregister, and pause a component at runtime. The difference between pausing and deregistering a component is as follows. If a component deregisters, then the engine clears all the associated data and references to this component; other components cannot synchronize with the deregistered component unless it registers anew. If a component is paused, other components cannot synchronize with it but the engine keeps all associated data and references to it; the paused component can start synchronizing with other components by simply informing the engine that it is back on track. %Notice that a component is paused only from the engine's perspective; this does not prevent it from executing or interacting with other components in a purely asynchronous manner. 

The rest of the report is structured as follows. Section \ref{secn:javabip} presents the JavaBIP framework. Section \ref{secn:example} describes our motivating case study. Section \ref{secn:validity} presents the notion of JavaBIP system validity and the construction of validity graphs. Section \ref{secn:interactionlogic} presents the interaction logic and the macro-notation used to specify JavaBIP synchronization constraints on component types. Section \ref{secn:implementation} describes the implemented software architecture and presents performance results. Section \ref{secn:related} discusses related work.  Section \ref{secn:conclusion} summarizes the results and future work directions.

\section{The JavaBIP Framework}
\label{secn:javabip}

JavaBIP implements the BIP (Behavior-Interaction-Priority) coordination mechanism~\cite{bip},
for coordination of concurrent components. In BIP, the behavior of components is described by Finite State Machines (FSMs) having transitions labeled with \emph{ports} and extended with data stored in local variables. Ports form the interface of a component and are used to define its interactions with other components.  They can also export part of the local variables, allowing access to the component's data. 
Component coordination is defined in BIP by means of \emph{interaction models}, \ie sets of interactions. Interactions are sets of ports that define allowed synchronizations among components.    

JavaBIP takes as input the \textit{system specification}, which is provided by the user and consists of the following:
\begin{itemize}
\item A \textit{behavior specification} for each component type, which is an FSM extended with ports and data provided as an annotated Java class.
\item The \textit{glue specification}, which is the interaction model of the system, is provided as an XML file. It specifies how the transitions of different component types must be synchronized, \ie synchronization constrains.
\item The optional \emph{data-wire specification}, which is the data transfer model of the system, is provided as an XML file. It specifies which and how data are exchanged among component types.  
\end{itemize}

For property analysis, the system specification can be automatically translated into an equivalent model of the system in the BIP language. This model can then be verified for deadlock freedom or other properties, using DFinder~\cite{dfinder}, ESST or nuXmv~\cite{esst4bip}.  Other analyses can be performed using any tool for which a model transformation from BIP is available.

\section{Motivating Case Study}
\label{secn:example}

Modular phones require application layer specifications that can handle dynamic device insertion and removal at runtime. In the rest of the paper, we refer to the phone's devices as \textit{modules}. In this case study, we model in JavaBIP some of the application layer protocols offered by Google's Greybus specification\footnote{https://github.com/projectara/greybus-spec}. 

\begin{figure}[t]
  \centering
  \includegraphics[width=0.8\textwidth]{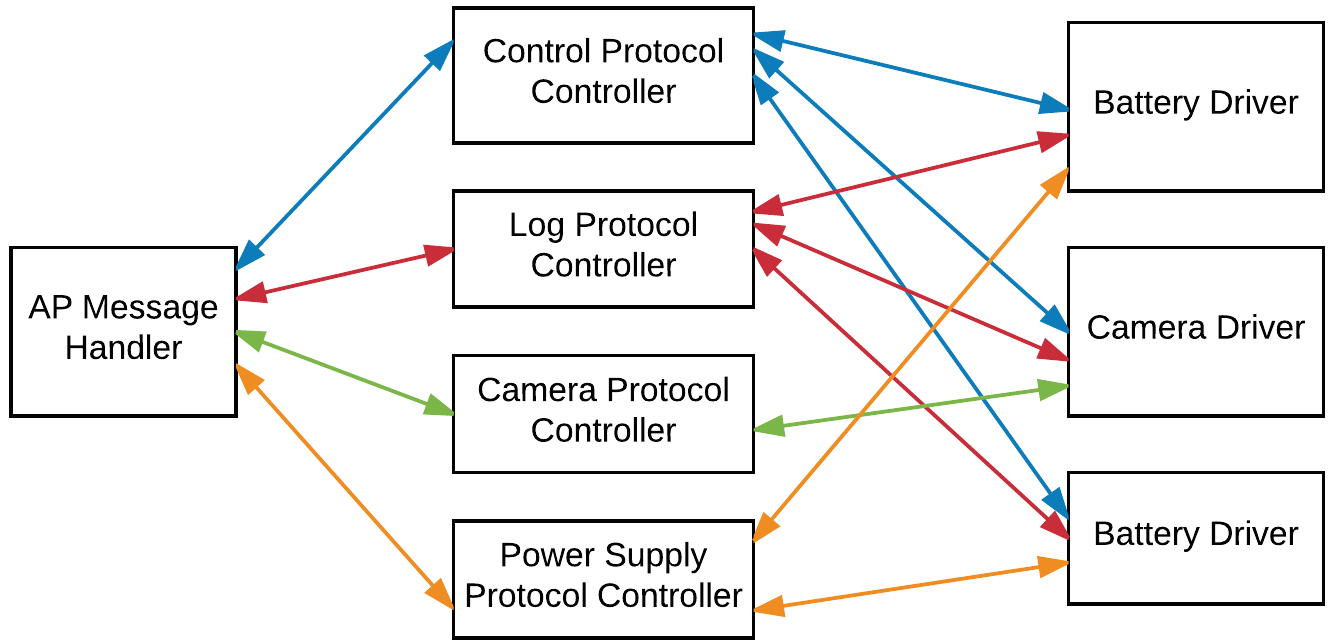}
  \caption{Modular phone architecture.}
  \label{fig:modularphone}
\end{figure}

\Fig{modularphone} illustrates the composite component types, of the case study. Greybus requires that exactly one application processor (AP) is present in the system for storing user data and executing applications. We consider two types of modules that can be inserted on the phone's frame at runtime: 1) power supply modules, \eg batteries and 2) cameras. Any number of instances of these two types can be inserted or removed from the phone at runtime. \Fig{modularphone} presents an example configuration of a phone, in which two battery and one camera modules are connected. These modules communicate with the AP through dedicated device class connection protocols: the \emph{camera}, \emph{power supply}, and \emph{log} protocols. The latter can be used by any module to send human-readable debug log messages to AP. Additionally, AP uses the \emph{control protocol} to perform basic initialization and configuration actions with other modules. 
\newcommand{\modifier}[1]{\ensuremath{\mathtt{#1}}} % Could be \mathit{#1} or other
\newcommand{\AppLayer}{\modifier{Application\ Layer}}%
\newcommand{\MessHand}{\modifier{AP\ Message\ Handler}}%
\newcommand{\Contr}{\modifier{Controller}}%
\newcommand{\Driver}{\modifier{Driver}}%
\newcommand{\ReqWorker}{\modifier{AP\ Request\ Worker}}%
\newcommand{\RespWorker}{\modifier{AP\ Response\ Worker}}%
\newcommand{\MessWorker}{\modifier{AP\ Message\ Worker}}%
\newcommand{\RecFifo}{\modifier{AP\ Receiver\ Fifo}}%
\newcommand{\CPC}{\modifier{Control\ Protocol\ Controller}}%
\newcommand{\LPC}{\modifier{Log\ Protocol\ Controller}}%
\newcommand{\CamPC}{\modifier{Camer\ Protocol\ Controller}}%
\newcommand{\PSPC}{\modifier{Power\ Supply\ Protocol\ Controller}}%
\newcommand{\BatDriver}{\modifier{Battery\ Driver}}%
\newcommand{\CamDriver}{\modifier{Camera\ Driver}}%
\newcommand{\CCH}{\modifier{Control\ Connect\ Handler}}%
\newcommand{\CDH}{\modifier{Control\ Disconnect\ Handler}}%
\newcommand{\LogH}{\modifier{Log\ Handler}}%
\newcommand{\CamCaptH}{\modifier{Camera\ Capture\ Handler}}%
\newcommand{\CamStreamH}{\modifier{Camera\ Stream\ Handler}}%
\newcommand{\PSH}{\modifier{Power\ Supply\ Handler}}%
\begin{figure} [t]
\[
  \renewcommand{\arraystretch}{1.5}
  \begin{array}{@{}l@{}}
    \begin{multlined}
      \AppLayer ::= (\MessHand).(\Contr)^+.(\Driver)^*
    \end{multlined}
    \\
    \begin{multlined}
      \MessHand ::= (\ReqWorker).(\RespWorker).\\[-12pt]
      (\MessWorker).(\RecFifo)
    \end{multlined}
    \\
    \begin{multlined}
      \Contr ::= (\CPC).(\LPC).\\[-12pt]
      (\CamPC).(\PSPC)
    \end{multlined}
    \\
    \begin{multlined}
      \Driver ::= (\BatDriver)^*.(\CamDriver)^*
    \end{multlined}
    \\
    \begin{multlined}
      \CamDriver ::= (\CCH).(\CDH).\\[-12pt]
      (\LogH).(\CamCaptH).(\CamStreamH)
    \end{multlined}
    \\
    \begin{multlined}
      \BatDriver ::= (\CCH).(\CDH).\\[-12pt]
      (\LogH).(\PSH)
    \end{multlined}
  \end{array}
\]
  \caption{Hierarchical decomposition of the \texttt{Application Layer} into components.}
  \label{fig:componentization}
\end{figure}
If no power supply or camera modules are connected, the system configuration would consist of the \texttt{AP Message Handler}, \texttt{Control Protocol Controller}, the \texttt{Log Protocol Controller}, \texttt{Camera Protocol Controller}, and the \texttt{Power Supply Protocol Controller} composite components. The grammar in \fig{componentization} shows how to obtain the desired systems as the incremental composition of components. Operators $.$ (dot), $\cdot^*$ and $\cdot^+$ are used as usual to denote composition and repetition. Notice that \fig{modularphone} illustrates only one of the possible system configurations that are described by the grammar in \fig{componentization}. 

A Greybus protocol defines a number of Greybus operations, which are \emph{request-response pairs} of remote procedure calls from one module to another. The bi-directional arrows in \fig{modularphone} represent Greybus operations.  For instance, the AP very often needs to retrieve information from other modules. This requires that a message requesting information be paired with a response message containing the information requested. In many cases, Greybus operations need to be performed in a specific order. Additionally, the access to shared resources such as memory and logging services needs to be controlled among modules. We enforce action flow of Greybus operations, as well as controlled access to the phone's shared resources with JavaBIP. We developed the case study using the WebGME-BIP design studio\footnote{https://github.com/anmavrid/webgme-bip}, the complete system exceeds $2000$ lines of code.

\subsection{Componentization and Interaction Model}
We have used \emph{architecture diagrams}~\cite{mavridou2016diagrams} to model the architecture style of the case study, which is shown in \fig{style}. An architecture style defines the set of all possible architectures for any number of components in the system. An architecture diagram consists of a set of \emph{component types}, with associated cardinality constraints  representing the expected number of instances of each component type and a set of \emph{connector motifs}.  The boxes in \fig{style} represent the atomic component types of the case style, which are the following: \ReqWorker, \RespWorker, \MessWorker, \CPC, \RecFifo, \LPC, \CamPC,  \CCH, \PSPC, \CDH, \LogH, \CamCaptH, \CamStreamH, \PSH. The cardinalities of these component types are shown in \fig{style} in the upper left corner of the corresponding boxes. For instance, the cardinality of the \ReqWorker\ component type is $1$, while the cardinality of the \CamCaptH\ component type is $n$. The valuation of the $n$ parameter depends on the number of cameras attached on the phone.

\Fig{style} contains a set of connector motifs, which define the glue of the case study, \ie the interaction model. Each connector motif defines a set of BIP connectors~\cite{bip}, which are non-empty sets of \emph{port types}.  Each connector motif end has two associated constraints: {\em multiplicity} and {\em degree}, represented as a pair $m : d$. Multiplicity constrains the number of instances of the port type that must participate in a connector defined by the motif; degree constrains the number of connectors attached to any instance of the port type. Cardinalities, multiplicities and degrees are either natural numbers or intervals. In this case study we have used only natural numbers.

Let us consider, for instance, the connector motif that connects the port type \texttt{receive} of \MessWorker\ with the port type \texttt{rm} of \RecFifo. The associated constraint pair of each connector motif end is equal to $1:1$. This means a conforming architecture of the style will include a binary BIP connector attached to the port instance \texttt{receive} of the component instance of \MessWorker~and to the port instance \texttt{rm} of the component instance of \RecFifo. This binary BIP connector represents a synchronization of the actions \texttt{rm} and \texttt{receive} of the corresponding component instances. 
Lets us now consider the connector motif that connects the port type \texttt{send\_log} of \LogH~with the port type \texttt{rcvFromDriver} of \LPC. The associated constraint pair of the connector motif end attached to \texttt{send\_log} is equal to $1:1$, while the constraint pair of the connector motif end attached to \texttt{rcvFromDriver} is equal to $1:n$, where $n$ is the cardinality of \LogH. A conforming architecture that contains $n$ instances of \LogH~will also contain $n$ binary connectors, each connecting a distinct component instance of \LogH~to the unique component instance of \LPC.  %% In general, a connector motif can define several possible sets of connectors for a given set of port instances.  In the present case study, this ambiguity is always resolved by using a refined port typing, based on additional action semantics, which we omit for the sake of the presentation simplicity.

\begin{sidewaysfigure}
  \centering
  \includegraphics[width=\textwidth]{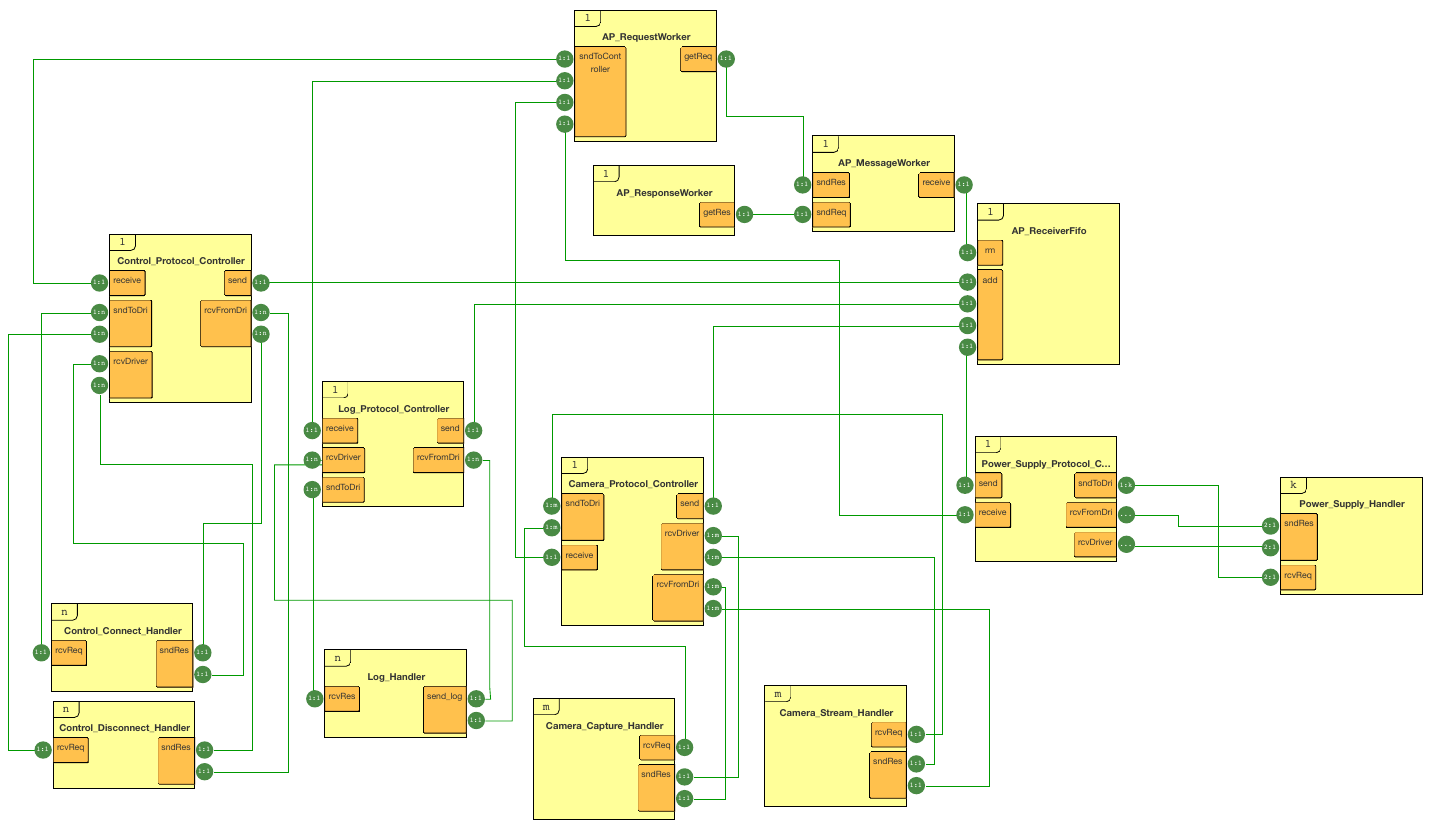}
  \caption{Modular phone architecture style.}
  \label{fig:style}
\end{sidewaysfigure}

\section{Interaction Logic and Macro-notation}
\label{secn:interactionlogic}

The glue specification is defined in JavaBIP through a macro-notation, similar to the one introduced in \cite{bozga2012modeling}, based on first-order interaction logic. This notation imposes synchronization constraints based on component types rather than on component instances, which allows a developer to write a glue specification without knowing the exact number of components in the system. Instances of component types for which synchronization constraints exist in the glue specification can be dynamically registered or deregistered at runtime without requiring additional input or changes in the glue specification. 
In the next subsections, we present the propositional and first-order interaction logic, as well as the JavaBIP Require/Accept macro-notation, which is based on first-order interaction logic. %The presented macro-notation is more expressive than the one presented in the static JavaBIP implementation~\cite{MiSE14p25} as it allows restricting the cardinality of a port type, \ie how many instances of this port type must synchronize in the same interaction.

\subsection{Propositional Interaction Logic}
\label{sec:ai:inter}
The propositional interaction logic (PIL), studied in \cite{alcon,BliSif10-causal-fmsd}, is a Boolean logic used to characterize a \emph{configuration}, \ie a non-empty set of interactions among components on a global set of ports $P$.  
We assume that $P$ is given and $a \subseteq P$.

\begin{definition}
An interaction $a$ is a set of ports $a \subseteq P$ such that $a \neq \emptyset$.
\end{definition}

\begin{syntax}
The propositional interaction logic is defined by the grammar:
\begin{align}
\phi ::= \true\ |\ p\ |\ \compi \phi\ |\ \phi \vee \phi
\,,
\qquad\text{with any $p \in P$.}
\end{align}
Conjunction  is defined as follows:  $\phi_1 \wedge \phi_2 \bydef{=} \compi{(\compi{\phi_1} \vee \compi{\phi_2})}$. Implication is defined as follows: $\phi_1 \Rightarrow \phi_2 \bydef{=} \compi{\phi_1} \vee \phi_2$.
To simplify the notation, we omit conjunction in monomials, \eg writing $sr_1r_2$ instead of $s \wedge r_1 \wedge r_2$.
\end{syntax}

\begin{semantics}
The meaning of a PIL formula $\phi$ is defined by the following satisfaction relation. Let $\gamma$ be a non-empty configuration. We define: $\gamma \models \phi$ iff for all $a \in \gamma$, $\phi$ evaluates to $\true$ for the valuation induced by $a$: $p=\true$, for all $p\in a$ and $p=\false$, for all $p \not\in a$.
\end{semantics}

The operators meet the usual Boolean axioms and the additional axiom $\bigvee_{p \in P} p = \true$  meaning that interactions are non-empty sets of ports.

\begin{example}
\label{ex:starArchitecture}
Consider a \emph{Star architecture}, where a single component $C$ acts as the center, and three other components $S_1,\ S_2,\ S_3$ communicate with the center through binary synchronizations. Component $C$ has a port $p$ and all other components have a single port $q_i\ (i=1,2,3)$. The corresponding PIL formula is: $p q_1\compi{q_2}\compi{q_3} \vee p \compi{q_1} q_2\compi{q_3} \vee p \compi{q_1} \compi{q_2} q_3$.
\end{example}

\subsection{First-order Interaction Logic}
\label{sec:il:firstorder}

We extend the propositional interaction logic presented in Subsection \ref{sec:ai:inter} with quantification over components to define interactions independently from the number of component instances.  This extension is particularly useful because, in practice, systems are built from multiple component instances of the same component type. A first-order interaction logic was also presented in~\cite{bozga2012modeling} with additional history variables. 
We make the following assumptions:
\begin{itemize}
\item A finite set of component types $\cT = \{T_1, \dots, T_n\}$ is given. Instances of a component type have the same interface and behavior. We write $c \oftype T$ to denote that a component $c$ is of type $T$. %Finally, we assume the existence of the universal component type U, such that any component is of this type.
\item The interface of each component type has a distinct set of ports. We denote by $T.p$ the \emph{port type} $p$, \ie a port belonging to the interface of type $T$. We write $T.P$ to denote the set of port types of component type $T$ and $\cT.P$ to denote the set of port types of all component types.
 We write $c.p$, for a component $c \oftype T$,  to denote the {\em port instance} of type $T.p$. %and $c.P$ to denote the set of ports of the component $c$. 
 \end{itemize}

Let $\phi$ denote any formula in PIL.

\begin{syntax}
The first-order interaction logic (FOIL) is defined by the grammar:
\begin{align}
\Phi ::= \true\ |\ \phi\ |\ \compi \Phi |\ \Phi \vee \Phi\ |\ \exists c:T\big(Pr(c) \big).\Phi
\,,
\end{align}
where $T$ is a component type, which represents a set of component instances with identical interfaces and behaviour. Variable $c$ ranges over component instances and must occur in the scope of a quantifier. $Pr(c)$ is some set-theoretic predicate on $c$ (omitted when $Pr =true$).
\end{syntax}

Additionally, we define the usual notation for the universal quantifier:
\[\forall c\oftype T \bigl(\Pr(c)\bigr). \Phi \bydef{=} \not \exists c\oftype T \bigl(\Pr(c)\bigr). \compi \Phi.\]

\begin{semantics}
The semantics is defined for closed formulas, where, for each variable in the formula, there is a quantifier over this variable in a higher nesting level. We assume that the finite set of component types $\cT = \{T_1,\dots,T_n\}$ is given. Models are pairs $\langle \mathcal{B}, \gamma\rangle$, where $\mathcal{B}$ is a set of component instances of types from $\cT$ and $\gamma$ is a configuration on the set of ports $P$ of these components. For quantifier-free formulas, the semantics is the same as for PIL formulas. For formulas with quantifiers, the satisfaction relation is defined as follows:
\begin{align*}
\langle \mathcal{B},\gamma \rangle &\models \exists c:T \big( Pr(c) \big).\Phi\,,
&&\text{iff  $\gamma \models \bigvee_{c':T \in B \wedge Pr(c')} \Phi[c'/c]$,}
\end{align*}
where $c':T$ ranges over all component instances of type $T \in \cT$ and $\Phi[c'/c]$ is obtained by replacing all occurrences of $c$ in $\Phi$ by $c'$.
\end{semantics}

\begin{figure}
  \centering
  \includegraphics[scale=1]{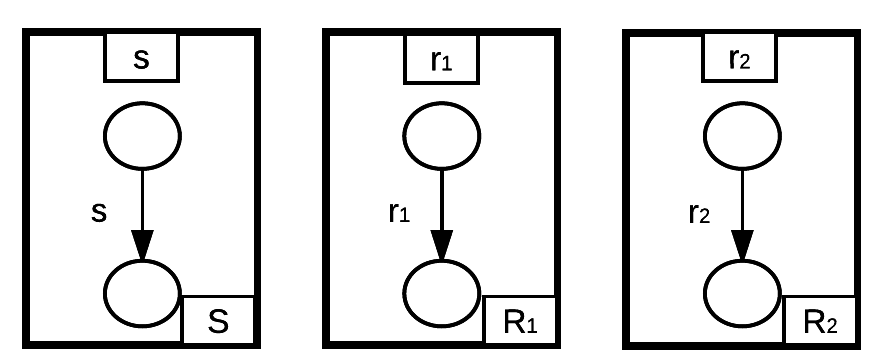}
  \caption{Three BIP components}
  \label{fig:bgexample}
\end{figure}

\begin{example}
\label{ex:bgexample}
Consider three components shown in \fig{bgexample}: a sender $S$ and two receivers $R_1,\ R_2$. The sender has a port $s$ for sending messages and each receiver has a port $r_i\ (i=1,2)$ for receiving them. We consider the following four coordination schemes:
\begin{itemize}
\item {\em Rendezvous} ensures strong synchronization between $S$ and all $R_i$. Rendezvous is specified by a single interaction involving all ports represented by the monomial $sr_1r_2$. This interaction can occur only if all of the components are in states enabling transitions labelled, respectively, by $s$, $r_1$ and $r_2$.
\item {\em Broadcast} allows all interactions involving $S$ and any (possible empty) subset of $R_i$. Broadcast is represented by the formula $s$, which can be expanded to $s\compi{r_1}\,\!\compi{r_2} \vee sr_1\compi{r_2} \vee s\compi{r_1}\!r_2 \vee sr_1r_2$. These interactions can only occur if $S$ is in a state enabling $s$. Each $R_i$ participates in an interaction only if it is in a state enabling $r_i$.
\item {\em Atomic broadcast} ensures that either all or none of the receivers are involved in the interaction. Atomic broadcast is be characterised by the formula $s\compi{r_1}\,\!\compi{r_2} \vee sr_1r_2$. The $s r_1 r_2$ interaction corresponds to a strong synchronization among the sender and all receivers.
\end{itemize} 
\end{example}

\begin{example}
\label{ex:javabip:star}
The Star architecture (\ex{starArchitecture}) can be expressed in FOIL for any number of components of type $S$ as follows:
\begin{center}
 $ \exists c \oftype C .\ \forall s \oftype S .\ (c.p\ s.q) \wedge\ \forall s' \oftype S\ (s \neq s').\ (\compi{s.q\ s'.q}) \wedge \forall c': C (c = c').\ true$.
\end{center}
\end{example}

\subsection{JavaBIP Require/Accept Macro-notation Based on FOIL}
\label{secn:model:macro}

JavaBIP relies on component types, rather than on component instances for the definition of synchronization constraints. All instances of a given component type are restricted with the same set of synchronization constraints. 

Consider a port $p$ of a component type $T$, which labels one or more transitions of $T$. The associated synchronization constraint to all transitions of $T$ labeled by $p$ is the conjunction of two constraints: the \emph{causal} and \emph{acceptance} constraints. Similarly to~\cite{bozga2012modeling}, two macros are used: 1)~the \textbf{Require} macro and 2)~the \textbf{Accept} macro to define the causal and acceptance constraints, respectively. Next, we describe the meaning of the two macros through examples. 

\subsubsection{The Require Macro} is used to specify ports required for synchronization. Let $T^1, T^2 \in \cT$ be two component types. The following:
\begin{align*}
  T_1.p \require T_2.q 
  \enskip \equiv \enskip
  \forall c_1 \oftype T_1.\ \exists c_2 \oftype T_2.\ \forall c_3 \oftype T_2\ (c_2 \ne c_3).\ \big( c_1.p\ \Rightarrow\ c_2.q\ \non{c_3.q} \big)\,,
\end{align*}

means that, to participate in an interaction, each of the ports $p$ of component instances of type $T_1$ requires synchronization with \emph{precisely one} of the ports $q$ of component instances of type $T_2$. In comparison with~\cite{bozga2012modeling}, we have opted for a macro-notation where the cardinality is explicit: should two instances of the same port type be required, this is specified by explicitly putting the required port type twice:
\begin{align*}
  T_1.p \require T_2.q\ T_2.q
  \enskip \equiv \enskip
  \forall c_1 \oftype T_1.\ \exists c_2, c_3 \oftype T_2.\ \forall c_4 \oftype T_2\ (c_2 \ne c_3 \ne c_4).\\ \big( c_1.p\ \Rightarrow\ c_2.q\ c_3.q\ \non{c_4.q} \big)\,,
\end{align*}

and so on for higher cardinalities. We call \emph{effect} what is specified in the left-hand side of \require (\eg $T_1.p$) and \emph{cause} what is specified in the right-hand side (\eg $T_2.q\ T_2.q$). A cause consists of a set of \emph{OR-causes}, where each OR-cause is a set of ports. For $p$ to participate in an interaction, all the ports belonging to at least one of the OR-causes must synchronize. We define: 

\begin{align*}
  T_1.p \require &T_2.q\ T_2.q\ ;\ T_2.r
  \enskip \equiv \enskip
  \forall c_1 \oftype T_1.\\ &\Big(\exists c_2, c_3 \oftype T_2.\ \forall c_4 \oftype T_2\ (c_2 \ne c_3 \ne c_4).\ \big( c_1.p\ \Rightarrow\ c_2.q\ c_3.q\ \non{c_4.q} \big)\\ &\bigvee\ \exists c_2 \oftype T_2.\ \forall c_3 \oftype T_2 (c_2 \ne c_3).\ \big( c_1.p\ \Rightarrow\ c_2.r\ \non{c_3.r} \big)\ 
  \Big)\,,
\end{align*}

which means that $p$ requires either the synchronization of two instances of $q$ or one instance of $r$. Notice the semicolon that separates the two OR-causes.

\subsubsection{The Accept Macro} defines optional ports for synchronization, \ie it defines the boundary of interactions. This is expressed by explicitly excluding from interactions all the ports that are not optional. Let $T^1, T^2 \in \cT$ be two component types. The following:

\begin{align*}
	\label{eq:accept}
	 T_1.p \accept T_2.q  \quad
	\equiv \quad \forall c_1 \oftype T_1.\ \left(\bigwedge_{T.r \in \cT.P \setminus \{T_2.q\}}  \forall c \oftype T.\ (c_1.p \Rightarrow \non{c.r}) \right)\,,
\end{align*}

means that $p$ accepts he synchronization of instances of $q$ but does not accept instances of any other port types.

The generalization of the above definitions to more complex macros is straightforward, but cumbersome.  Therefore we omit it here.

\begin{example}
\label{ex:starMacros}
The synchronization constraints of the Star architecture (\ex{javabip:star}) are specified by the following combination of macros:
\begin{align*}
  \mathtt{S.q} &\require \mathtt{C.p}
  &\mathtt{S.q} &\accept \mathtt{C.p}
  \\
  \mathtt{C.p} &\require \mathtt{S.q}
  &\mathtt{C.p} &\accept \mathtt{S.q}
  &&,
\end{align*}
\end{example}

\begin{example}
\label{ex:general}
Let us now consider a more general example to illustrate the expressiveness of the synchronization constraints.  Assume that there are five component types \texttt{A, B, C, D, E} with port types \texttt{a, b, c, d, e}, respectively. Through the Require macros, we enforce the following five constraints: 1)~\texttt{A.a} requires synchronization with two instances of \texttt{B.b}; 2)~\texttt{B.b} requires synchronization either with a)~a single instance of \texttt{A.a} and a single instance of \texttt{C.c} or b)~just two instances of \texttt{C.c}; 3)~\texttt{C.c} does not require synchronizations with other ports (however it accepts synchronisations with any possible combination of ports \texttt{A.a, B.b, C.c}); 4)~\texttt{D.d} requires synchronization with a single instance of \texttt{E.e} and 5)~\texttt{E.e} does not require synchronizations with other ports (however it accepts synchronisations with any number of ports \texttt{D.d}). Notice that by the combination of the first two require macros, a synchronisation involving exactly an instance of \texttt{A.a} and two instances of \texttt{B.b} is not allowed, since \texttt{B.b} requires at least one instance of \texttt{C.c} to also participate in the synchronisation.

{\tabsize
\begin{align*}
  \mathtt{A.a} &\require \mathtt{B.b}\ \mathtt{B.b}
  &\mathtt{A.a} &\accept \mathtt{A.a}\ \mathtt{B.b}\ \mathtt{C.c}
  \\
  \mathtt{B.b} &\require \mathtt{A.a}\  \mathtt{C.c}\  ;\ \mathtt{C.c}\ \mathtt{C.c}
  &\mathtt{B.b} &\accept \mathtt{A.a}\ \mathtt{B.b}\ \mathtt{C.c}
  \\
 \mathtt{C.c} &\require  -
  &\mathtt{C.c} &\accept \mathtt{A.a}\  \mathtt{B.b}\ \mathtt{C.c}
  \\
    \mathtt{D.d} &\require \mathtt{E.e}
  &\mathtt{D.d} &\accept \mathtt{E.e}
  \\
    \mathtt{E.e} &\require -
  &\mathtt{E.e} &\accept \mathtt{D.d}
\end{align*}
}%
\label{generalExampleGlue}
\end{example}

\section{Defining System Validity}
\label{secn:validity}

In the previous, static JavaBIP implementation, a developer would first register all components to the engine and then start the engine manually. Since, in the dynamic JavaBIP implementation, components may register or deregister on the fly, we need a notion of validity so that depending on whether there are enough registered components, the engine can automatically start or stop its execution. We start by formally defining atomic components, BIP systems and valid BIP systems.

\begin{definition} [Component] 
\label{defn:atomicComponent}
A component $B$ is a Finite State Machine represented by a triple $(Q, P, \goesto[])$, where $Q$ is a set of states, $P$ is a set of communication ports, $\goesto[] \subseteq Q \times P \times Q$ is a set of transitions, each labeled by a port.
\end{definition}

Below, we use the common notation, writing $q \goesto[p] q'$
instead of $(q,p,q') \in\, \goesto{}$.

\begin{definition} [BIP System]
\label{defn:system}
A BIP system is defined by a composition operator parameterized by a set of interactions $\gamma \subseteq 2^P$. $\cB_n = \gamma(B_1,\dots,B_n)$ is a Finite State Machine $(Q,\gamma,\goesto[])$, where $Q = \prod_{i=1}^n Q_i$ and $\goesto[]$ is the least set of transitions satisfying the following rule:
\[ \frac{a = \{p_i\}_{i \in I} \in \gamma\qquad \forall i \in I: q_i \goesto[p_i] q_i'\qquad \forall i \notin I : q_i = q'_i}{(q_1, \dots, q_n) \goesto[a] (q'_1, \dots, q'_n)} \]
\end{definition}

The inference rule says that a BIP system, consisting of $n$ components, can execute an interaction $a \in \gamma$, iff for each port $p_i \in a$, the corresponding component $B_i$, can execute a transition labeled with $p_i$; the states of components that do not participate in the interaction remain the same. The set of possible interactions of a BIP system is defined in JavaBIP by the glue specification, \ie the set of Require and Accept macros.
 We write $B \oftype T$ to denote a component $B$ of type $T$.  We denote by $\cT$ the set of all component types of a BIP system.

\Defn{validitySystem} extends \defn{system} to describe a valid BIP system. System validity is defined from the perspective of starting or stopping the JavaBIP engine, which orchestrates component interaction (synchronization of component actions). Notice that even if a system is not valid according to \defn{validitySystem}, JavaBIP components can be communicating in an asynchronous manner.

\begin{definition} [Valid BIP System]
\label{defn:validitySystem}
A BIP system $(Q,\gamma,\goesto[])$ is valid iff $\gamma \neq \emptyset$. \end{definition}

\begin{remark}
  In \defn{atomicComponent} and \defn{system}, for the sake of simplicity, we omit the presentation of data-related aspects. However, it should be noted that the full JavaBIP~\cite{SPE:SPE2495} allows data variables within components. In such cases, component transitions can be guarded by Boolean predicates on data variables. Notice that in \defn{validitySystem} we do not consider guards. This is a design choice that we made. The result of guard evaluation might easily change multiple times throughout the system lifecycle, \eg based on the components internal state or on component interaction. Thus, it is undesirable to base engine execution on such often recurring changes, which could actually result in increasing the engine's overhead.
\end{remark}

\Defn{validitySystem} says that a BIP system is valid if and only if there are enough registered components such that the interaction set of the system is not empty. %Notice that \defn{validitySystem} does not consider transition guard evaluation. 
To determine the validity of a system, we use directed graphs with edge coloring to model dependencies among components. The generation of the validity graph is based on the Require macros of the glue specification, since these define the minimum number of required interactions among the components. The complete glue specification is used by the engine for orchestrating component execution.. 

\begin{definition} [Validity graph]
\label{defn:validityGraph}
A labelled graph $G = (\cT, E, c)$ is the validity graph of a set of Require macros iff:
\begin{enumerate}
\item~the vertex set $\cT$ is the set of component types defined in the Require macros;
\item~the edge set $E$ contains a directed edge $(T_1,T_2)$ iff there exists a Require macro that contains $T_1$ in the effect and $T_2$ in an OR-cause;
\item for each edge $(T_1,T_2) \in E$, the counter $c: E \rightarrow \mathbb{Z}$ is initialized with the cardinality of $T_2$ in the corresponding OR-cause. 
\end{enumerate}
The edges of the graph are colored such that: 1) all edges corresponding to an OR-cause of a Require macro are colored the same; 2) edges corresponding to different OR-causes are colored differently.  \end{definition}

Clearly, there always exists a validity graph for any set of Require macros. Note that the outgoing edges of two different vertices may have the same color.

\begin{wrapfigure}{r}{0.4\textwidth}
    \centering
    \includegraphics[width=0.27\textwidth]{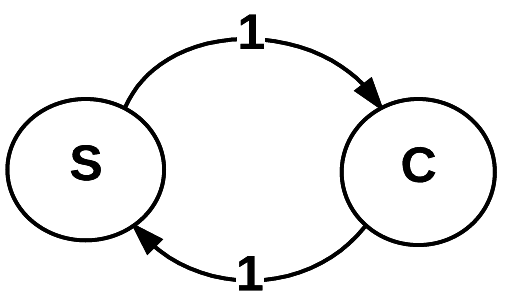}
    \caption{Validity graph of \ex{starMacros}.}
\label{fig:starValidity}
\end{wrapfigure}

\fig{starValidity} shows the validity graph of the Star architecture defined in the glue specification of \ex{starMacros}. The two vertices represent the two component types \texttt{S, D} of the example. Since, component type \texttt{S} requires synchronization with exactly one instance of component type \texttt{C} (cardinality = $1$), there is an edge from vertex \texttt{S} to vertex \texttt{C} labeled by a counter initialised to $1$. Analogously, since, component type \texttt{C} requires synchronization with exactly one instance of component type \texttt{S} (cardinality = $1$), there is an edge from vertex \texttt{C} to vertex \texttt{S} labeled by a counter initialised to $1$.

\begin{wrapfigure}{r}{0.4\textwidth}
    \centering
    \includegraphics[width=0.4\textwidth]{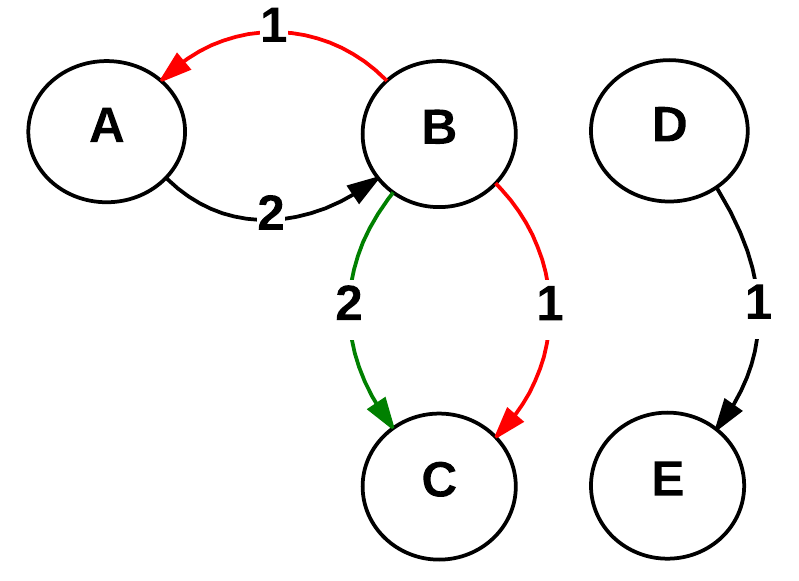}
    \caption{Validity graph of \ex{general}.}
\label{fig:example_system}
\end{wrapfigure}

\fig{example_system} shows the validity graph that corresponds to the glue specification of \ex{general}. There are five vertices \texttt{A, B, C, D, E} each representing a component type of the example. Let us consider the second and third Require macros of the example. In the second Require macro notice the two OR-causes, each represented by edges with different colors (red and green) in \fig{example_system}. In the first OR-cause, \texttt{B} requires one instance of \texttt{A} and one instance of \texttt{C}, represented by the red edges in \fig{example_system} both labeled by $1$ since both cardinalities are equal to $1$. In the second OR-cause, \texttt{B} requires two instances of \texttt{C}, represented by the green edge in \fig{example_system} labeled by $2$. In the third Require macro of the example, \texttt{C} does not require synchronization with any other component of the system, and thus, there is no outgoing edge from vertex \texttt{C}.

\begin{figure} [t]
    \centering
    \includegraphics[width=\textwidth]{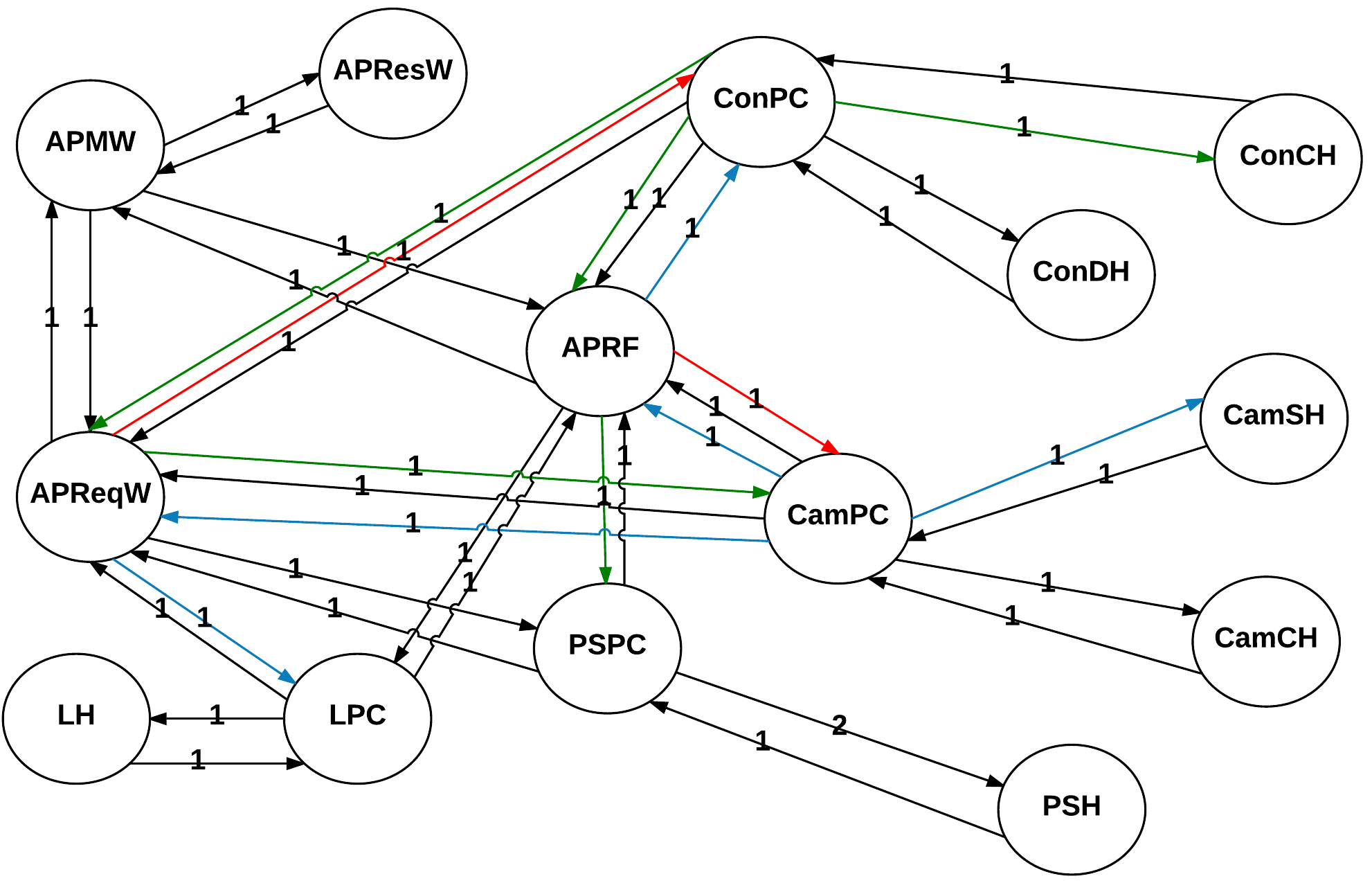}
    \caption{Validity graph of the Modular Phone case study.}
\label{fig:phoneValidity}
\end{figure}

\fig{phoneValidity} shows  the validity graph of the modular phone case study.  The are $13$ vertices, each one representing an atomic component type of the case study (\fig{componentization}). These are the following: \texttt{AP Request Worker}, \texttt{AP Response Worker}, \texttt{AP Message Worker}, \texttt{AP Receiver Fifo}, \texttt{Control Protocol Controller}, \texttt{Log Protocol Controller}, \texttt{Camera Protocol Controller}, \texttt{Power Supply Protocol Controller}, \texttt{Control Connect Handler}, \texttt{Control Disconnect Handler}, \texttt{Log Handler}, \texttt{Camera Capture Handler}, \texttt{Camera Stream Handler}, and \texttt{Power Supply Handler}. Due to space limitations, we have substituted the full names with their acronyms. For instance, we have substituted \texttt{AP Message Handler} by \texttt{APMW}. In case of acronym conflicts, we added more letters, \eg we have substituted \texttt{AP Request Worker} by \texttt{APReqW}, and \texttt{AP Response Worker} by \texttt{APResW}.

Let us now consider two of the Require macros of the case study (the full set of Require and Accept macros can be found in the Appendix).\\

\texttt{PSPC.rcvFromDri} \require \texttt{PSH.sndRes} \texttt{PSH.sndRes}

\texttt{APRF.add} \require \texttt{ConPC.snd}; \texttt{CamPC.snd}; \texttt{PSPC.snd}; \texttt{LPC.snd}\\

 Since, component type \texttt{Power Supply Protocol Controller} requires synchronization with two instances of component type \texttt{Power Supply Handler}, there is an edge from vertex \texttt{PSPC} to vertex \texttt{PSH} labeled by a counter initialized to $2$. Furthermore, component type \texttt{AP Receiver Fifo} requires synchronization either with an instance of component type \texttt{Control Protocol Controller} or an instance of component type \texttt{Camera Protocol Controller} or an instance of \texttt{Power Supply Protocol Controller} or an instance of \texttt{Logger Protocol Controller}. Thus, there are four outgoing edges from vertex \texttt{APRF}, each labeled by a counter initialized to $1$ and colored by a different color, to vertices \texttt{ConPC}, \texttt{CamPC}, \texttt{PSPC}, and \texttt{LPC}, respectively. In \fig{phoneValidity}, edges with different colors  are also represented by different line styles.

\begin{definition} [Dynamic change in validity graph]
\label{defn:dynamicChange}
In the event of a dynamic change, a validity graph is updated as follows:
\begin{enumerate}
\item If a component instance of type \texttt{T} is registered, the counters of all incoming edges of vertex \texttt{T} are decremented by $1$.
\item If a component instance of type \texttt{T} is deregistered or paused, the counters of all incoming edges of vertex
 \texttt{T} are incremented by $1$.
\end{enumerate}
\end{definition}

\begin{proposition} [Determining system validity]
\label{prop:validity}
Consider a BIP system and a corresponding validity graph. The BIP system is valid iff for at least one vertex of the validity graph, an instance of the vertex's corresponding type is registered and the counters of all outgoing edges of at least one color are equal to or less than $0$.
\end{proposition}

\begin{proof}
Necessity: Since $\gamma \neq \emptyset$, there exists at least one interaction $a \in \gamma$. This means that there exists a \texttt{Require} macro such that the effect is the type of a port instance $p \in a$ and the causes contain an OR-cause $OR_p$ that consists of a subset of the types of port instances in $a$. We denote this subset by $Q$. For each port type $q \in Q$, there exists at least one instance of $q \in a$. The cardinality of each port type in $OR_p$ is equal to the number of corresponding instances in $a$. By Definition \ref{defn:validityGraph}, there exists a vertex $v$ in the validity graph that corresponds to the component type that contains $p$ and a set of outgoing edges. Let us denote $V_p$ the subset of outgoing edges of $v$ that correspond to $OR_p$, which are colored the same. Since $a \in \gamma$, this implies that all counters of $V_p$ are equal to or less than $0$ and there is at least a registered instance of the component type that contains $p$.

Sufficiency: We know that there exists a vertex with a registered component instance and a color such that the counters of all the outgoing edges of that color are equal to or less than $0$. Consider a \texttt{Require} macro with an OR-cause $OR_c$ that corresponds to this color. Since all the counters of edges of this color are equal to $0$, we know that there exists a sufficient number of registered instances of every type in $OR_c$ . Since the effect of the \texttt{Require} macro is also registered, there is an enabled interaction involving component instances corresponding to the effect and the elements of $OR_c$ . This implies that the system is valid, which concludes the proof. 
\end{proof}

\begin{example}
Let us now consider a system with the glue specification of \ex{general}. \fig{dynamicChanges} present the changes of the corresponding validity graph (see \fig{example_system}) when 1)~an instance of \texttt{B} is registered; 2)~an instance of \texttt{E} is registered; 3)~an instance of \texttt{C} is registered and 4)~the instance of \texttt{E} is deregistered. Notice that the system becomes valid when an instance of \texttt{E} is registered and becomes invalid when the instance of \texttt{E} is deregistered. \end{example}

\begin{figure}[t]
    \centering
    \includegraphics[width=\textwidth]{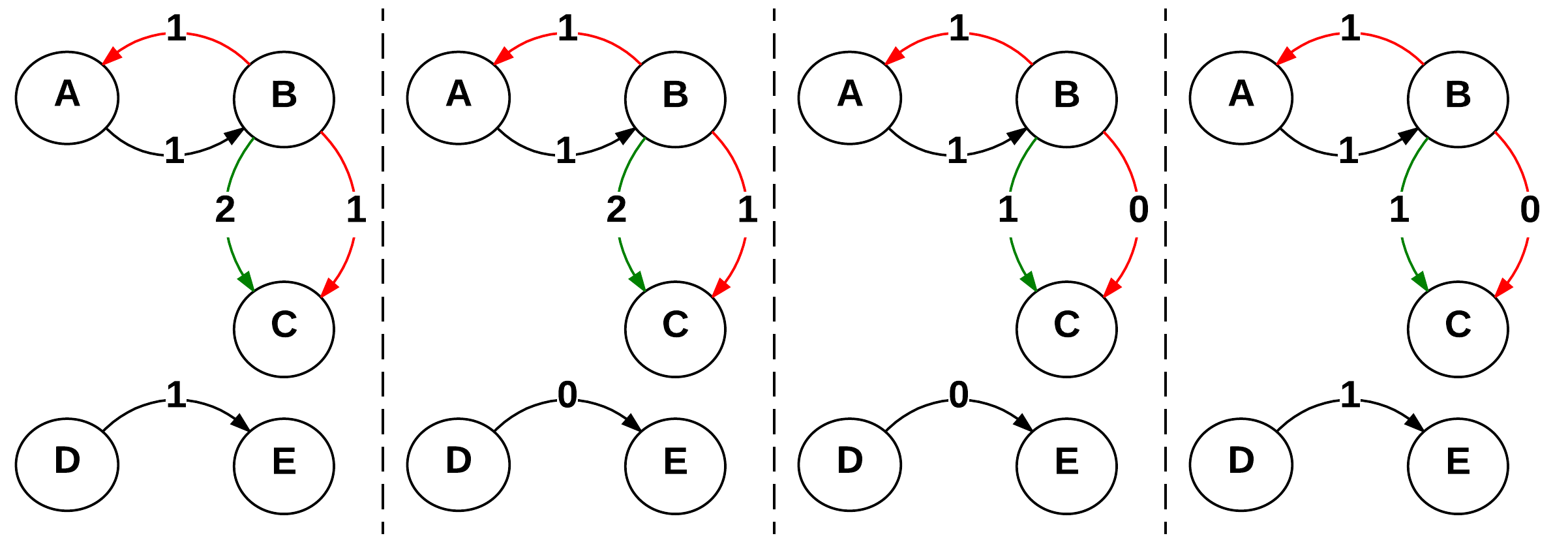}
    \caption{Changes in validity graph when adding/removing components.}
\label{fig:dynamicChanges}
\end{figure}

%\begin{example}
%\Fig{dynamicChanges} presents the changes in the validity graph of \fig{phoneValidity}, when at least one instance of each of the \texttt{APMW}, \texttt{APReqW}, \texttt{APResW}, \texttt{APRF}, and \texttt{PSH} component types are registered. Notice that the system is still not valid; to become valid it would require instances of all component types contained in at least one of the red dashed boxes to be registered. For instance, if a component instance of type \texttt{LPC} and a component instance of type \texttt{LH} get registered, the system becomes valid. The blue dotted boxes represent choice between components. Each blue dotted box contained in a single red dashed box corresponds to an OR-cause of the same Require macro. For instance, if a component instance of type \texttt{ConPC} get registered with a component instance of type \texttt{ConDH} or with a component instance of type\texttt{ConCH}, the system also becomes valid.
%\end{example}

%\begin{figure}[t]
%    \centering
%    \includegraphics[width=\textwidth]{addingComponentsPhone.png}
%    \caption{Changes in validity graph of the case study.}
%\label{fig:dynamicChanges}
%\end{figure}

To start and stop the engine, we determine first whether the system is valid by using Prop.~\ref{prop:validity}. Nevertheless, we do not need to check system validity every time a component registers/deregisters/pauses. Corollaries~\ref{cor:valid} and \ref{cor:invalid} define such cases. 

\begin{corollary}
\label{cor:valid}
If a BIP system  $\cB_n$ is valid and a component is registered, then the new BIP system $\cB_{n+1}$ is also valid.
\end{corollary}

\begin{proof}
It follows directly from Proposition \ref{prop:validity} and Definition \ref{defn:dynamicChange}.

\end{proof}

\begin{corollary} 
\label{cor:invalid}
If a BIP system $\cB_n$ is invalid and a component is deregistered or paused, then the new BIP system $\cB_{n-1}$ is also invalid.
\end{corollary}

\begin{proof}
It follows directly from Proposition \ref{prop:validity} and Definition \ref{defn:dynamicChange}.
\end{proof}

\section{Implementation}
\label{secn:implementation}

Next, we discuss the implementation of the dynamic JavaBIP extension, during which the implementation of the JavaBIP engine has significantly changed. 
Let us consider first the interface of the JavaBIP engine, \ie \texttt{BIPEngine}. In the static implementation, \texttt{BIPEngine} consisted of the following functions: 1)~\texttt{register} used by a developer to register a component to the engine; 2)~\texttt{inform} used by a component to inform the engine of its current state and enabled transitions; 3)~\texttt{specifyGlue} used by a developer to send the glue specification to engine; 4)~\texttt{start} used by a developer to start the engine thread and 5)~\texttt{stop} used by a developer to stop the engine.

We updated \texttt{BIPEngine} as follows. Function \texttt{start} was removed, since the engine thread is now started automatically based on whether enough components are registered to form a valid system. We added two functions: 1)~\texttt{deregister} used by a developer or the component itself (\eg in the case of a failure) to deregister from the engine and 2)~\texttt{pause} used by a developer or the component (\eg in the case that the component is going to communicate asynchronously with other components for an amount of time) to pause synchronizations with other components. Function \texttt{register} was considerably updated, as well as function \texttt{stop} which can also be called internally by the engine in the case of an invalid system. The remaining functions were been updated. \Fig{architecture} shows the software architecture of the JavaBIP engine. The arrows labeled \texttt{register}, \texttt{deregister}, \texttt{stop}, \texttt{specifyGlue}, and \texttt{pause} represent calls to the \texttt{BIPEngine} functions.

\begin{figure}
    \centering
    \includegraphics[width=0.9\textwidth]{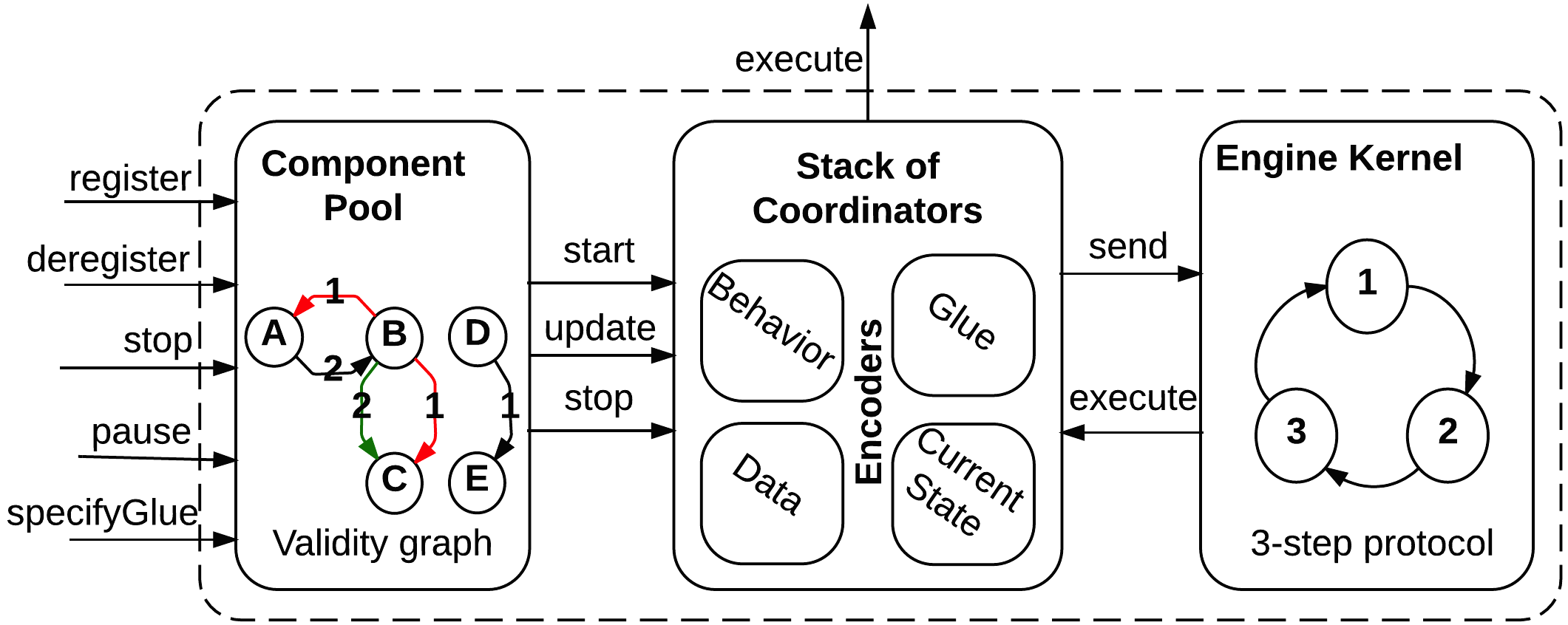}
    \caption{Dynamic JavaBIP Engine software architecture.}
\label{fig:architecture}
\end{figure}

The \texttt{ComponentPool} object was added, which is used as an interface to the validity graph described in \defn{validityGraph}. The \texttt{ComponentPool} starts the core engine (comprising a stack of coordinators and the engine kernel), when the system becomes valid, and stops it, when the system becomes invalid. System validity is checked whenever a component is registered, deregistered or paused, excluding the cases described in Cor.~\ref{cor:valid} and \ref{cor:invalid}. Whenever a component is registered or deregistered without affecting the validity of the system, the \texttt{Component Pool} sends an update registration/deregistration event to the core engine.

The engine composes and solves the various constraints of the system. Its implementation is based on Binary Decision Diagrams (BDDs)~\cite{bdd}, which are efficient data structures to store and manipulate Boolean formulas.\footnote{We have used the JavaBDD package, available at \url{http://javabdd.sourceforge.net}} The imposed constraints encode information about the behavior, glue, data, and current state of the components. Current state constraints allow us to compute the enabled transitions of the component. For each type of constraints, we discuss which parts must be recomputed when registering components at runtime. There is no need to recompute these constraints when a component is paused or deregistered. Whenever constraints are recomputed, the Coordinators send these to the kernel. 

The formulas that define the behavior, glue, data, and current state constraints were presented in~\cite{SPE:SPE2495}. \Fig{constraints} summarizes the constraint computation. The white color indicates that the constraint is computed only once at system initialization. The light gray indicates that the constraint is recomputed when a component is registered. The dark gray color indicates that the constraint is recomputed during each execution cycle.

\begin{figure}
    \centering
    \includegraphics[width=0.65\textwidth]{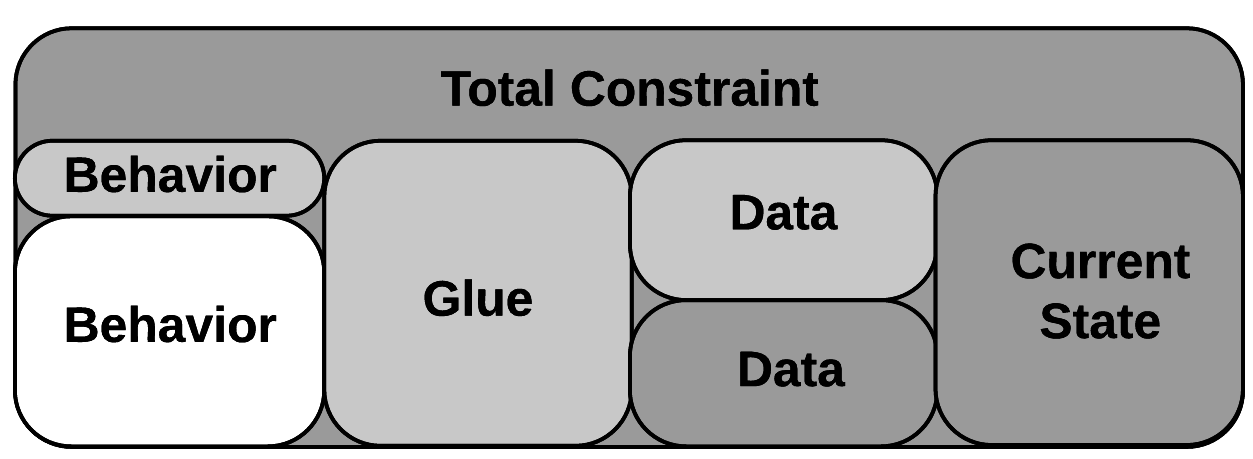}
    \caption{Constraint computation phases.}
\label{fig:constraints}
\end{figure}

The \emph{behavior constraint} of a component includes the ports and states of the component. For each port, a Boolean port variable is created. Similarly, for each state, a Boolean state variable is created. Behavior constraints are built using these port and state variables. %Let $Q_i$ and $P_i$ be the set of states and ports of a component $C_i$. Let $P_i^q=\setdef{p \in P_i}{ \exists\ q':\, q \goesto{p} q'}$. We compute the behavior BDD for component $C_i$ as follows:
%\[
%\mathit{BDD}^B_i = \bigvee_{q \in Q_i} \left( q \land \bigwedge_{q' \in Q_i\setminus\{q\}} \overline{q'} \land \bigvee_{p \in P^q_i} \left( p \land %\bigwedge_{p \notin P_i^q} \overline{p}   \right) \right) \vee \bigwedge_{p \in P_i} \overline{p}
%\,.
%\]
The \emph{total behavior constraint} is computed as the conjunction of all component behavior constraints. When a component is registered, its behavior constraint is computed and conjuncted to the total behavior constraint. When a component is deregistered, its port variables are set to $\mathit{false}$.
%
%\[
%\mathit{BDD}^B = \bigwedge_{i=1}^n \mathit{BDD}^B_i.
%\]

The \emph{glue constraint} is computed by interpreting the Require and Accept macros of the glue specification. The same Boolean port variables that were previously created for the behavior constraints are used for the glue constraint as well. The glue constraint must be recomputed, in a valid system, every time a new component is registered.  

For the \emph{data constraint}, additional data variables have to be created. The data constraints represent how data is exchanged among components, \ie which components are providing data and which components are consuming data. For each pair of components exchanging data, a data variable is created. When a component is registered, the data constraints that involve the newly arrived components are recomputed. Components exchange data at the beginning of each execution cycle of the system. Based on the exchanged data, components may disable some of the possible interactions. As a result, a subset of data constraints is recomputed at each execution cycle.

The \emph{current state constraint} of a component is computed when a component informs of its disabled transitions due to guard evaluation. %Ports can be disabled due to evaluation of guards. %Let $q$ be the current state and $P_i^{\mathit{dis}}$ be the set of disabled ports of a component $C_i$. Then, the current state BDD for component $C_i$ is computed as follows:
%\[
%\mathit{BDD}^Q_i = q \land \bigwedge_{q'\neq q,\, q' \in Q_i} \overline{q'} \land \bigwedge_{p \in P_i^{\mathit{dis}}} \overline{p}\,.
%\]
The \emph{total current state BDD} is the conjunction of the current state constraints of all registered components. During engine execution, \ie in a valid system, the total current state constraint is computed at each execution cycle of the engine and is further conjuncted with the total behavior constraint, the glue constraint, and the total data constraint. 

The execution of a JavaBIP valid system is driven by the engine kernel applying the following protocol in a cyclic manner: 
\begin{enumerate}
\item Upon reaching a state, all component constraints are sent to the kernel;
\item The kernel computes the \emph{total constraint}, which is the conjunction of the total behavior, glue, current state and data constraints. Thus, it computes the possible interactions satisfying the total system constraint and picks one of them;
\item The kernel notifies the Coordinators of its decision by calling \texttt{execute}, which then notify the components to execute the necessary transitions.
\end{enumerate}

Notice that a component can be registered during any step of the engine protocol. The engine, however will only include the newly registered component in the BDD computation at the beginning of the next cycle. System validity is checked, when a component is paused or deregistered. If the system remains valid and the engine is executing the second or third step of the engine protocol, the engine sets the port variables of this component to $\mathit{false}$ and recomputes the possible interactions.

\subsection{Performance Results}
\label{secn:performance}

We show performance results for the modular phone case study. The experiments were performed on a 3.1 GHz Intel Core i7 with 8GB RAM. We started with $5$ registered components and registered up to $45$ additional components. The JavaBIP models are available online\footnote{https://github.com/sbliudze/javabip-itest}. Table \ref{my-label} summarizes the engine's computation times and the BDD Manager peak memory usage for various numbers of components. We present and discuss three different engine times: 1) the time needed to perform a complete engine execution cycle (three-step protocol run by the Engine kernel); 2) the time needed to (partially) recompute the behavior, glue, and data BDD constraints due to the registration of a new component; 3) the time needed to add or remove a component from the component tool and check the validity of the system.

The first column shows the number of components in the system, after the registration or deregistration of a component. For instance, $10$ means that a new component was registered and the total number of components in the system is now $10$. The number of components is also decreased in two cases, when it is equal to $11$ and equal to $29$. This means that a component was deregistered or paused and the total number of components in the system is $11$ or $29$, respectively.  

The second column shows the average engine execution time of the first $1000$ engine cycles after a component registration or deregistration. The system becomes valid and the engine is started upon the registration of the $12^{\mathrm{th}}$ component. As a result, the engine execution times are equal to $0$ for the first two rows of the table. If the engine had been started, for instance, after the registration of the $5^{\mathrm{th}}$ component (without the system being valid), the engine would have needed $<1$~ms per execution cycle. This means that an overhead of seconds or minutes could have been added in the system's execution if more than a certain number of engine execution cycles (\eg $100000$) had been performed by the time the system became valid. 

 The third column of Table \ref{my-label} shows the amount of time needed to recompute the behavior, glue, and data constraints of the system due to a component registration. The first two rows are equal to $0$ since the system is invalid and thus, no BDD computation is required. 
If the engine had been started before the system became valid, the BDDs would have been recomputed upon the registration of each new component. For instance, after the registration of the $5^{\mathrm{th}}$ component, the engine would have needed $13$~ms and after the registration of the $11^{\mathrm{th}}$ component, the engine would have needed additional $49$~ms to recompute the BDDs. The fifth column shows the peak memory usage of the BDD manager after a component registration or deregistration.

Finally, the fourth column of Table \ref{my-label} presents the amount of time needed to add or remove a component from the component pool and check for system validity. The time needed is very low, in some cases even less than $1$ ms. These were the cases that system validity was not checked due to the results of Cor.~\ref{cor:valid} and \ref{cor:invalid}. The system became valid when the $12^{\mathrm{th}}$ component was registered. This required the maximum amount of time ($3.654$~ms), since the full graph was checked for validity, and then the core engine thread was started. Next, a component was deregistered, the system became invalid again, and the engine thread was stopped. The amount of time needed by the component pool was $2.908$~ms.

\begin{table}[]
\centering
\caption{Engine times and BDD Manager peak memory usage. Times are in milliseconds and memory usage is in Megabytes.}
\label{my-label}
\begin{tabular}{c|c|c|c|c}
\hline
\begin{tabular}[c]{@{}c@{}}Number of\\ components\end{tabular} & \begin{tabular}[c]{@{}c@{}}Time: Engine\\ execution cycle\end{tabular} & \begin{tabular}[c]{@{}c@{}}Time: BDD\\ (re)computation\end{tabular} & \begin{tabular}[c]{@{}c@{}}Time:\\ Component pool\end{tabular} & Memory \\ \hline
5                                                          & 0                                                            & 0                                                               & 2.078                                                   & 0  \\
10                                                         & 0                                                            & 0                                                                & 2.186                                                   & 0  \\
12                                                         & \textless 1                                                            & 63                                                                & 3.654                                                              & 0.059  \\
11                                                         & 0                                                            & 0                                                                 & 2.908                                                              & 0.057  \\
20                                                         & \textless 1                                                            & 151                                                               & \textless 1                                                    & 0.083  \\
25                                                         & 1.149                                                                  & 194                                                               & \textless 1                                                    & 0.099  \\
30                                                         & 1.247                                                                  & 239                                                               & \textless 1                                                    & 0.129  \\
29                                                         & 1.241                                                                  & 0                                                                 & 2.451                                                             & 0.121  \\
40                                                         & 1.399                                                                  & 283                                                               & \textless 1                                                    & 0.199  \\
50                                                         & 1.896                                                                  & 337                                                               & \textless 1                                                    & 0.254  \\ \hline
\end{tabular}
\end{table}

\section{Related Work}
\label{secn:related}

Dynamicity in BIP has been studied by several authors~\cite{bozga2012modeling,bruni14-bip-dynamics,Giusto2011}. In~\cite{bozga2012modeling}, the authors present the Dy-BIP framework that allows dynamic reconfiguration of connectors among the ports of the system.  They use {\em history variables} to allow sequences of interactions with the same instance of a given component type. JavaBIP can emulate history variables using data. In contrast, our focus is on dynamicity due to the creation and deletion of components that is often encountered in modern software systems that are not restricted to the embedded systems domain. Additionally, the interface-based design and the modular software architecture of JavaBIP allow us to easily extend the JavaBIP implementation.

Our approach is closest to~\cite{bruni14-bip-dynamics} and~\cite{edelmann17-FunctBIP}.  In~\cite{bruni14-bip-dynamics}, two extensions of the BIP model are defined: reconfigurable\mdash similar to Dy-BIP\mdash and dynamic, allowing dynamic replication of components.  They focus on the operational semantics of the two extensions and their properties, by studying their encodability in BIP and Place/Transition Petri nets (P/T Nets).  Composition is defined through interaction models, without considering structured connectors. In contrast, our work focuses mostly on the connectivity among components, defined by Require/Accept relations.
In~\cite{edelmann17-FunctBIP}, the BIP coordination mechanisms are implemented by a set of connector combinators in Haskell and Scala. Functional BIP provides combinators for managing connections in a dynamically evolving set of components.  However, as in~\cite{bruni14-bip-dynamics}, such evolution must be managed by explicit actions of existing components.  In contrast, the JavaBIP approach allows components to be created independently, only requiring that they be subsequently registered with the JavaBIP engine.

The Reo coordination language~\cite{Papadopoulos2001}\mdash
which realizes component coordination through circuit-like
connectors built of channels and nodes\mdash provides
dedicated primitives for reconfiguring connectors by
creating new channels ($\mathit{Ch}$), and manipulating
channel ends and nodes (\texttt{split}, \texttt{join},
\texttt{hide} and \texttt{forget}).  A number of papers
study reconfiguration of Reo connectors.  In particular,
\cite{Clarke08-ReconfLogic} provides a framework for model
checking reconfigurable circuits, whereas
\cite{Koehler08-dataflow} and \cite{Krause11-dyReo} take the
approach based on graph transformation techniques.  The main
difference between connector reconfiguration in Reo and
dynamicity in JavaBIP is that, in Reo, reconfiguration
operations are performed on constituent elements of the
connector.  Thus, in principle, such operations can affect
\emph{ongoing} interactions.  This is not possible in
JavaBIP, since interactions are completely atomic.

Three main types of formalisms have been studied in the literature for the specification of dynamic architectures and architecture styles~\cite{Bradbury:2004}: 1) graph grammars, 2) process algebras, and 3) logics. Graph grammars have been used to specify reconfiguration in a dynamic architecture through the use of graph rewriting rules. Representative approaches include the Le M\'etayer approach~\cite{le1998describing}, where nodes plus CSP-like behavior specifications are components and edges are connectors. A different way of representing software architectures with graph grammars can be found in~\cite{hirsch1998graph}, where hyperedges with CCS labels are components and nodes are communication ports. Other graph-based approaches are summarized in~\cite{BRUNI200839}. None of these approaches offers tool support. 

Additionally, process algebras have been used to define dynamic architectures as part of several architecture description languages (ADLs). For instance, $\pi$-calculus~\cite{milner1992calculus} was used in Darwin~\cite{magee1996dynamic} and LEDA~\cite{canal1999specification}, CCS was used in PiLar~\cite{cuesta2001dynamic}, and CSP was used in Dynamic Wright~\cite{allen1998specifying}. In comparison with our approach, Darwin and PiLar support only binary bindings (connectors), while in Dynamic Wright and LEDA there is no clear distinction between behavior and coordination since connectors can have behavior. 
%
%In~\cite{Giusto2011}, the authors revisit the BIP expressiveness, by introducing simple behaviour, such as prefixing, in the BIP glue operators.  They show that such minor modifications can rapidly lead to Turing completeness of glue.  In contrast to~\cite{Giusto2011} and other frameworks, such as~\cite{inverardi1995formal,le1998describing}, JavaBIP relies on a clear distinction between behavior and coordination.

Logic has also been used for the specification of dynamic software architectures and architecture styles. Alloy's first-order logic~\cite{jackson2002alloy} was used in~\cite{georgiadis2002self} for the specification of dynamic architectures, while the Alloy Analyzer tool was used to analyze these specifications. JavaBIP specifications can also be analyzed~\cite{dfinder,esst4bip}, however, the main focus of JavaBIP is runtime coordination, which is not offered in~\cite{georgiadis2002self}. Configuration logics~\cite{mavridou2017configuration} were proposed for the specification of architecture styles, which however, in their current form do not capture dynamic change.

\section{Conclusion and Future Work}
\label{secn:conclusion}

We presented an extension of the JavaBIP framework for coordination of software components that can register, deregister and pause at runtime. To handle this type of dynamicity, JavaBIP uses a macro-notation based on first-order interaction logic that allows specifying synchronization constraints on component types. This way, a developer is not required to know the exact number of components that need to be coordinated when specifying the synchronization constraints of a system. Additionally, we introduced a notion of system validity, which we use to start and stop the JavaBIP engine automatically at runtime depending on whether there are enough registered components in the system so that there is at least one enabled synchronization among them. In the previous, static JavaBIP implementation, developers had to manually start and stop the engine. Starting and stopping the engine in an automatic way helps optimize JavaBIP performance since it eliminates the engine's overhead in the case of an invalid system.

JavaBIP implements the principles of the BIP component framework rooted in rigorous operational semantics. Notice, however, that none of the current BIP engine implementations can handle dynamic insertion and deletion of components at runtime. The functionality of pausing a component at runtime makes the implementation of the JavaBIP engine more incremental. In our previous, static implementation, the engine had to wait for all registered components to inform in each cycle before making any computations. As a result, a single component could introduce a long delay in the system execution. In the current implementation, when a component is paused, the engine does not wait for it to inform, but rather computes the set of enabled interaction in the system that involve only the non-paused components. JavaBIP is an open-source tool\footnote{github.com/sbliudze/javabip-core, github.com/sbliudze/javabip-engine}. 

In the future, we plan to work towards increasing the incrementality of the engine in the following way: the engine does not have to wait for all non-paused components to inform but rather checks whether there is an enabled interaction among the components that have already informed and orders its execution. To check the enableness of interactions we plan to reuse the notion of validity graphs introduced in this paper and extend it with additional information on component ports. Additionally, we plan on extending the engine functionality to handle registration of new component types and synchronization patterns.
%To be mentioned in future work: 1)~registering batches of components; 2)~registering new types and 3)~engine incrementality.

\begin{appendices}

\section{Complete glue specification of the modular phone case study}

The glue specification must be provided in an XML file.  Each constraint has two parts: \texttt{effect}  and \texttt{causes}.  The former defines the port to which the  constraint is associated\mdash intuitively, the effect is the firing  of a transition labeled by this port.  The latter lists the ports  that are necessary to ``cause'' the ``effect''.  For the  \texttt{require} constraints, all causes must be present; for the
 \texttt{accept} constraints, any (possibly empty) combination of the  causes is accepted. \texttt{require} constraints have a set of options for causes, \ie the OR-causes explained in Subsection 3.3.

\begin{lstlisting}[style=customxml]
<?xml version="1.0" encoding="UTF-8"?>
<glue>
  <accepts>
    <accept>
      <effect id="sndToDri" specType="Log_Protocol_Controller"/>
      <causes>
        <port id="rcvRes" specType="Log_Handler"/>
      </causes>
    </accept>
    <accept>
      <effect id="rcvFromDri" specType="Log_Protocol_Controller"/>
      <causes>
        <port id="send_log" specType="Log_Handler"/>
      </causes>
    </accept>
    <accept>
      <effect id="send" specType="Log_Protocol_Controller"/>
      <causes>
        <port id="add" specType="AP_ReceiverFifo"/>
      </causes>
    </accept>
    <accept>
      <effect id="rcvDriver" specType="Log_Protocol_Controller"/>
      <causes>
        <port id="send_log" specType="Log_Handler"/>
      </causes>
    </accept>
    <accept>
      <effect id="receive" specType="Log_Protocol_Controller"/>
      <causes>
        <port id="sndToController" specType="/f/e/Z"/>
      </causes>
    </accept>
    <accept>
      <effect id="add" specType="AP_ReceiverFifo"/>
      <causes>
        <port id="send" specType="Control_Protocol_Controller"/>
        <port id="send" specType="Camera_Protocol_Controller"/>
        <port id="send" specType="Log_Protocol_Controller"/>
        <port id="send" specType="Power_Supply_Protocol_Controller"/>
      </causes>
    </accept>
    <accept>
      <effect id="rm" specType="AP_ReceiverFifo"/>
      <causes>
        <port id="receive" specType="AP_MessageWorker"/>
      </causes>
    </accept>
    <accept>
      <effect id="sndRes" specType="AP_MessageWorker"/>
      <causes>
        <port id="getReq" specType="/f/e/Z"/>
      </causes>
    </accept>
    <accept>
      <effect id="receive" specType="AP_MessageWorker"/>
      <causes>
        <port id="rm" specType="AP_ReceiverFifo"/>
      </causes>
    </accept>
    <accept>
      <effect id="sndReq" specType="AP_MessageWorker"/>
      <causes>
        <port id="getRes" specType="AP_ResponseWorker"/>
      </causes>
    </accept>
    <accept>
      <effect id="sndToController" specType="AP_RequestWorker"/>
      <causes>
        <port id="receive" specType="Control_Protocol_Controller"/>
        <port id="receive" specType="Log_Protocol_Controller"/>
        <port id="receive" specType="Power_Supply_Protocol_Controller"/>
        <port id="receive" specType="Camera_Protocol_Controller"/>
      </causes>
    </accept>
    <accept>
      <effect id="getReq" specType="AP_RequestWorker"/>
      <causes>
        <port id="sndRes" specType="AP_MessageWorker"/>
      </causes>
    </accept>
    <accept>
      <effect id="getRes" specType="AP_ResponseWorker"/>
      <causes>
        <port id="sndReq" specType="AP_MessageWorker"/>
      </causes>
    </accept>
    <accept>
      <effect id="sndToDri" specType="Control_Protocol_Controller"/>
      <causes>
        <port id="rcvReq" specType="Control_Disconnect_Handler"/>
        <port id="rcvReq" specType="Control_Connect_Handler"/>
      </causes>
    </accept>
    <accept>
      <effect id="rcvFromDri" specType="Control_Protocol_Controller"/>
      <causes>
        <port id="sndRes" specType="Control_Connect_Handler"/>
        <port id="sndRes" specType="Control_Disconnect_Handler"/>
      </causes>
    </accept>
    <accept>
      <effect id="send" specType="Control_Protocol_Controller"/>
      <causes>
        <port id="add" specType="AP_ReceiverFifo"/>
      </causes>
    </accept>
    <accept>
      <effect id="rcvDriver" specType="Control_Protocol_Controller"/>
      <causes>
        <port id="sndRes" specType="Control_Connect_Handler"/>
        <port id="sndRes" specType="Control_Disconnect_Handler"/>
      </causes>
    </accept>
    <accept>
      <effect id="receive" specType="Control_Protocol_Controller"/>
      <causes>
        <port id="sndToController" specType="AP_RequestWorker"/>
      </causes>
    </accept>
    <accept>
      <effect id="sndToDri" specType="Power_Supply_Protocol_Controller"/>
      <causes>
        <port id="rcvReq" specType="Power_Supply_Handler"/>
      </causes>
    </accept>
    <accept>
      <effect id="rcvFromDri" specType="Power_Supply_Protocol_Controller"/>
      <causes>
        <port id="sndRes" specType="Power_Supply_Handler"/>
      </causes>
    </accept>
    <accept>
      <effect id="send" specType="Power_Supply_Protocol_Controller"/>
      <causes>
        <port id="add" specType="AP_ReceiverFifo"/>
      </causes>
    </accept>
    <accept>
      <effect id="rcvDriver" specType="Power_Supply_Protocol_Controller"/>
      <causes>
        <port id="sndRes" specType="Power_Supply_Handler"/>
      </causes>
    </accept>
    <accept>
      <effect id="receive" specType="Power_Supply_Protocol_Controller"/>
      <causes>
        <port id="sndToController" specType="AP_RequestWorker"/>
      </causes>
    </accept>
    <accept>
      <effect id="sndToDri" specType="Camera_Protocol_Controller"/>
      <causes>
        <port id="rcvReq" specType="Camera_Stream_Handler"/>
        <port id="rcvReq" specType="Camera_Capture_Handler"/>
      </causes>
    </accept>
    <accept>
      <effect id="rcvFromDri" specType="Camera_Protocol_Controller"/>
      <causes>
        <port id="sndRes" specType="Camera_Capture_Handler"/>
        <port id="sndRes" specType="Camera_Stream_Handler"/>
      </causes>
    </accept>
    <accept>
      <effect id="send" specType="Camera_Protocol_Controller"/>
      <causes>
        <port id="add" specType="AP_ReceiverFifo"/>
      </causes>
    </accept>
    <accept>
      <effect id="rcvDriver" specType="Camera_Protocol_Controller"/>
      <causes>
        <port id="sndRes" specType="Camera_Capture_Handler"/>
        <port id="sndRes" specType="Camera_Stream_Handler"/>
      </causes>
    </accept>
    <accept>
      <effect id="receive" specType="Camera_Protocol_Controller"/>
      <causes>
        <port id="sndToController" specType="AP_RequestWorker"/>
      </causes>
    </accept>
    <accept>
      <effect id="rcvReq" specType="Power_Supply_Handler"/>
      <causes>
        <port id="sndToDri" specType="Power_Supply_Protocol_Controller"/>
      </causes>
    </accept>
    <accept>
      <effect id="sndRes" specType="Power_Supply_Handler"/>
      <causes>
        <port id="rcvDriver" specType="Power_Supply_Protocol_Controller"/>
        <port id="rcvFromDri" specType="Power_Supply_Protocol_Controller"/>
      </causes>
    </accept>
    <accept>
      <effect id="rcvRes" specType="Log_Handler"/>
      <causes>
        <port id="sndToDri" specType="Log_Protocol_Controller"/>
      </causes>
    </accept>
    <accept>
      <effect id="send_log" specType="Log_Handler"/>
      <causes>
        <port id="rcvDriver" specType="Log_Protocol_Controller"/>
        <port id="rcvFromDri" specType="Log_Protocol_Controller"/>
      </causes>
    </accept>
    <accept>
      <effect id="rcvReq" specType="Control_Connect_Handler"/>
      <causes>
        <port id="sndToDri" specType="Control_Protocol_Controller"/>
      </causes>
    </accept>
    <accept>
      <effect id="sndRes" specType="Control_Connect_Handler"/>
      <causes>
        <port id="rcvFromDri" specType="Control_Protocol_Controller"/>
      </causes>
    </accept>
    <accept>
      <effect id="rcvReq" specType="Camera_Stream_Handler"/>
      <causes>
        <port id="sndToDri" specType="Camera_Protocol_Controller"/>
      </causes>
    </accept>
    <accept>
      <effect id="sndRes" specType="Camera_Stream_Handler"/>
      <causes>
        <port id="rcvDriver" specType="Camera_Protocol_Controller"/>
        <port id="rcvFromDri" specType="Camera_Protocol_Controller"/>
      </causes>
    </accept>
    <accept>
      <effect id="rcvReq" specType="Control_Disconnect_Handler"/>
      <causes>
        <port id="sndToDri" specType="Control_Protocol_Controller"/>
      </causes>
    </accept>
    <accept>
      <effect id="sndRes" specType="Control_Disconnect_Handler"/>
      <causes>
        <port id="rcvDriver" specType="Control_Protocol_Controller"/>
        <port id="rcvFromDri" specType="Control_Protocol_Controller"/>
      </causes>
    </accept>
    <accept>
      <effect id="rcvReq" specType="Camera_Capture_Handler"/>
      <causes>
        <port id="sndToDri" specType="Camera_Protocol_Controller"/>
      </causes>
    </accept>
    <accept>
      <effect id="sndRes" specType="Camera_Capture_Handler"/>
      <causes>
        <port id="rcvDriver" specType="Camera_Protocol_Controller"/>
        <port id="rcvFromDri" specType="Camera_Protocol_Controller"/>
      </causes>
    </accept>
  </accepts>
  <requires>
    <require>
      <effect id="sndToDri" specType="Log_Protocol_Controller"/>
      <causes>
        <option>
          <causes>
            <port id="rcvRes" specType="Log_Handler"/>
          </causes>
        </option>
      </causes>
    </require>
    <require>
      <effect id="rcvFromDri" specType="Log_Protocol_Controller"/>
      <causes>
        <option>
          <causes>
            <port id="send_log" specType="Log_Handler"/>
          </causes>
        </option>
      </causes>
    </require>
    <require>
      <effect id="send" specType="Log_Protocol_Controller"/>
      <causes>
        <option>
          <causes>
            <port id="add" specType="AP_ReceiverFifo"/>
          </causes>
        </option>
      </causes>
    </require>
    <require>
      <effect id="rcvDriver" specType="Log_Protocol_Controller"/>
      <causes>
        <option>
          <causes>
            <port id="send_log" specType="Log_Handler"/>
          </causes>
        </option>
      </causes>
    </require>
    <require>
      <effect id="receive" specType="Log_Protocol_Controller"/>
      <causes>
        <option>
          <causes>
            <port id="sndToController" specType="AP_RequestWorker"/>
          </causes>
        </option>
      </causes>
    </require>
    <require>
      <effect id="add" specType="AP_ReceiverFifo"/>
      <causes>
        <option>
          <causes>
            <port id="send" specType="Control_Protocol_Controller"/>
          </causes>
        </option>
        <option>
          <causes>
            <port id="send" specType="Camera_Protocol_Controller"/>
          </causes>
        </option>
        <option>
          <causes>
            <port id="send" specType="Log_Protocol_Controller"/>
          </causes>
        </option>
        <option>
          <causes>
            <port id="send" specType="Power_Supply_Protocol_Controller"/>
          </causes>
        </option>
      </causes>
    </require>
    <require>
      <effect id="rm" specType="AP_ReceiverFifo"/>
      <causes>
        <option>
          <causes>
            <port id="receive" specType="AP_MessageWorker"/>
          </causes>
        </option>
      </causes>
    </require>
    <require>
      <effect id="sndRes" specType="AP_MessageWorker"/>
      <causes>
        <option>
          <causes>
            <port id="getReq" specType="AP_RequestWorker"/>
          </causes>
        </option>
      </causes>
    </require>
    <require>
      <effect id="receive" specType="AP_MessageWorker"/>
      <causes>
        <option>
          <causes>
            <port id="rm" specType="AP_ReceiverFifo"/>
          </causes>
        </option>
      </causes>
    </require>
    <require>
      <effect id="sndReq" specType="AP_MessageWorker"/>
      <causes>
        <option>
          <causes>
            <port id="getRes" specType="AP_ResponseWorker"/>
          </causes>
        </option>
      </causes>
    </require>
    <require>
      <effect id="sndToController" specType="AP_RequestWorker"/>
      <causes>
        <option>
          <causes>
            <port id="receive" specType="Control_Protocol_Controller"/>
          </causes>
        </option>
        <option>
          <causes>
            <port id="receive" specType="Log_Protocol_Controller"/>
          </causes>
        </option>
        <option>
          <causes>
            <port id="receive" specType="Power_Supply_Protocol_Controller"/>
          </causes>
        </option>
        <option>
          <causes>
            <port id="receive" specType="Camera_Protocol_Controller"/>
          </causes>
        </option>
      </causes>
    </require>
    <require>
      <effect id="getReq" specType="AP_RequestWorker"/>
      <causes>
        <option>
          <causes>
            <port id="sndRes" specType="AP_MessageWorker"/>
          </causes>
        </option>
      </causes>
    </require>
    <require>
      <effect id="getRes" specType="AP_ResponseWorker"/>
      <causes>
        <option>
          <causes>
            <port id="sndReq" specType="AP_MessageWorker"/>
          </causes>
        </option>
      </causes>
    </require>
    <require>
      <effect id="sndToDri" specType="Control_Protocol_Controller"/>
      <causes>
        <option>
          <causes>
            <port id="rcvReq" specType="Control_Disconnect_Handler"/>
          </causes>
        </option>
        <option>
          <causes>
            <port id="rcvReq" specType="Control_Connect_Handler"/>
          </causes>
        </option>
      </causes>
    </require>
    <require>
      <effect id="rcvFromDri" specType="Control_Protocol_Controller"/>
      <causes>
        <option>
          <causes>
            <port id="sndRes" specType="Control_Connect_Handler"/>
          </causes>
        </option>
        <option>
          <causes>
            <port id="sndRes" specType="Control_Disconnect_Handler"/>
          </causes>
        </option>
      </causes>
    </require>
    <require>
      <effect id="send" specType="Control_Protocol_Controller"/>
      <causes>
        <option>
          <causes>
            <port id="add" specType="AP_ReceiverFifo"/>
          </causes>
        </option>
      </causes>
    </require>
    <require>
      <effect id="rcvDriver" specType="Control_Protocol_Controller"/>
      <causes>
        <option>
          <causes>
            <port id="sndRes" specType="Control_Connect_Handler"/>
          </causes>
        </option>
        <option>
          <causes>
            <port id="sndRes" specType="Control_Disconnect_Handler"/>
          </causes>
        </option>
      </causes>
    </require>
    <require>
      <effect id="receive" specType="Control_Protocol_Controller"/>
      <causes>
        <option>
          <causes>
            <port id="sndToController" specType="AP_RequestWorker"/>
          </causes>
        </option>
      </causes>
    </require>
    <require>
      <effect id="sndToDri" specType="Power_Supply_Protocol_Controller"/>
      <causes>
        <option>
          <causes>
            <port id="rcvReq" specType="Power_Supply_Handler"/>
          </causes>
        </option>
      </causes>
    </require>
    <require>
      <effect id="rcvFromDri" specType="Power_Supply_Protocol_Controller"/>
      <causes>
        <option>
          <causes>
            <port id="sndRes" specType="Power_Supply_Handler"/>
            <port id="sndRes" specType="Power_Supply_Handler"/>
          </causes>
        </option>
      </causes>
    </require>
    <require>
      <effect id="send" specType="Power_Supply_Protocol_Controller"/>
      <causes>
        <option>
          <causes>
            <port id="add" specType="AP_ReceiverFifo"/>
          </causes>
        </option>
      </causes>
    </require>
    <require>
      <effect id="rcvDriver" specType="Power_Supply_Protocol_Controller"/>
      <causes>
        <option>
          <causes>
            <port id="sndRes" specType="Power_Supply_Handler"/>
            <port id="sndRes" specType="Power_Supply_Handler"/>
          </causes>
        </option>
      </causes>
    </require>
    <require>
      <effect id="receive" specType="Power_Supply_Protocol_Controller"/>
      <causes>
        <option>
          <causes>
            <port id="sndToController" specType="AP_RequestWorker"/>
          </causes>
        </option>
      </causes>
    </require>
    <require>
      <effect id="sndToDri" specType="Camera_Protocol_Controller"/>
      <causes>
        <option>
          <causes>
            <port id="rcvReq" specType="Camera_Stream_Handler"/>
          </causes>
        </option>
        <option>
          <causes>
            <port id="rcvReq" specType="Camera_Capture_Handler"/>
          </causes>
        </option>
      </causes>
    </require>
    <require>
      <effect id="rcvFromDri" specType="Camera_Protocol_Controller"/>
      <causes>
        <option>
          <causes>
            <port id="sndRes" specType="Camera_Capture_Handler"/>
          </causes>
        </option>
        <option>
          <causes>
            <port id="sndRes" specType="Camera_Stream_Handler"/>
          </causes>
        </option>
      </causes>
    </require>
    <require>
      <effect id="send" specType="Camera_Protocol_Controller"/>
      <causes>
        <option>
          <causes>
            <port id="add" specType="AP_ReceiverFifo"/>
          </causes>
        </option>
      </causes>
    </require>
    <require>
      <effect id="rcvDriver" specType="Camera_Protocol_Controller"/>
      <causes>
        <option>
          <causes>
            <port id="sndRes" specType="Camera_Capture_Handler"/>
          </causes>
        </option>
        <option>
          <causes>
            <port id="sndRes" specType="Camera_Stream_Handler"/>
          </causes>
        </option>
      </causes>
    </require>
    <require>
      <effect id="receive" specType="Camera_Protocol_Controller"/>
      <causes>
        <option>
          <causes>
            <port id="sndToController" specType="AP_RequestWorker"/>
          </causes>
        </option>
      </causes>
    </require>
    <require>
      <effect id="rcvReq" specType="Power_Supply_Handler"/>
      <causes>
        <option>
          <causes>
            <port id="sndToDri" specType="Power_Supply_Protocol_Controller"/>
          </causes>
        </option>
      </causes>
    </require>
    <require>
      <effect id="sndRes" specType="Power_Supply_Handler"/>
      <causes>
        <option>
          <causes>
            <port id="rcvDriver" specType="Power_Supply_Protocol_Controller"/>
          </causes>
        </option>
        <option>
          <causes>
            <port id="rcvFromDri" specType="Power_Supply_Protocol_Controller"/>
          </causes>
        </option>
      </causes>
    </require>
    <require>
      <effect id="rcvRes" specType="Log_Handler"/>
      <causes>
        <option>
          <causes>
            <port id="sndToDri" specType="Log_Protocol_Controller"/>
          </causes>
        </option>
      </causes>
    </require>
    <require>
      <effect id="send_log" specType="Log_Handler"/>
      <causes>
        <option>
          <causes>
            <port id="rcvDriver" specType="Log_Protocol_Controller"/>
          </causes>
        </option>
        <option>
          <causes>
            <port id="rcvFromDri" specType="Log_Protocol_Controller"/>
          </causes>
        </option>
      </causes>
    </require>
    <require>
      <effect id="rcvReq" specType="Control_Connect_Handler"/>
      <causes>
        <option>
          <causes>
            <port id="sndToDri" specType="Control_Protocol_Controller"/>
          </causes>
        </option>
      </causes>
    </require>
    <require>
      <effect id="sndRes" specType="Control_Connect_Handler"/>
      <causes>
        <option>
          <causes>
            <port id="rcvFromDri" specType="Control_Protocol_Controller"/>
          </causes>
        </option>
        <option>
          <causes>
            <port id="rcvFromDri" specType="Control_Protocol_Controller"/>
          </causes>
        </option>
      </causes>
    </require>
    <require>
      <effect id="rcvReq" specType="Camera_Stream_Handler"/>
      <causes>
        <option>
          <causes>
            <port id="sndToDri" specType="Camera_Protocol_Controller"/>
          </causes>
        </option>
      </causes>
    </require>
    <require>
      <effect id="sndRes" specType="Camera_Stream_Handler"/>
      <causes>
        <option>
          <causes>
            <port id="rcvDriver" specType="Camera_Protocol_Controller"/>
          </causes>
        </option>
        <option>
          <causes>
            <port id="rcvFromDri" specType="Camera_Protocol_Controller"/>
          </causes>
        </option>
      </causes>
    </require>
    <require>
      <effect id="rcvReq" specType="Control_Disconnect_Handler"/>
      <causes>
        <option>
          <causes>
            <port id="sndToDri" specType="Control_Protocol_Controller"/>
          </causes>
        </option>
      </causes>
    </require>
    <require>
      <effect id="sndRes" specType="Control_Disconnect_Handler"/>
      <causes>
        <option>
          <causes>
            <port id="rcvDriver" specType="Control_Protocol_Controller"/>
          </causes>
        </option>
        <option>
          <causes>
            <port id="rcvFromDri" specType="Control_Protocol_Controller"/>
          </causes>
        </option>
      </causes>
    </require>
    <require>
      <effect id="rcvReq" specType="Camera_Capture_Handler"/>
      <causes>
        <option>
          <causes>
            <port id="sndToDri" specType="Camera_Protocol_Controller"/>
          </causes>
        </option>
      </causes>
    </require>
    <require>
      <effect id="sndRes" specType="Camera_Capture_Handler"/>
      <causes>
        <option>
          <causes>
            <port id="rcvDriver" specType="Camera_Protocol_Controller"/>
          </causes>
        </option>
        <option>
          <causes>
            <port id="rcvFromDri" specType="Camera_Protocol_Controller"/>
          </causes>
        </option>
      </causes>
    </require>
  </requires>
</glue>
\end{lstlisting}

\end{appendices}


\begin{thebibliography}{10}

\bibitem{actorsAgha}
Gul Agha.
\newblock {\em Actors: a model of concurrent computation in distributed
  systems}.
\newblock MIT Press, Cambridge, MA, USA, 1986.

\bibitem{bdd}
S.B. Akers.
\newblock Binary decision diagrams.
\newblock {\em IEEE Transactions on Computers}, C-27(6):509--516, 1978.

\bibitem{allen1998specifying}
Robert Allen, Remi Douence, and David Garlan.
\newblock Specifying and analyzing dynamic software architectures.
\newblock {\em Fundamental Approaches to Software Engineering}, pages 21--37,
  1998.

\bibitem{bip}
Ananda Basu, Saddek Bensalem, Marius Bozga, Jacques Combaz, Mohamad Jaber,
  Thanh-Hung Nguyen, and Joseph Sifakis.
\newblock Rigorous {C}omponent-{B}ased {S}ystem {D}esign {U}sing the {BIP}
  {F}ramework.
\newblock {\em {IEEE} {S}oftware}, 28(3):41--48, 2011.

\bibitem{bip06}
Ananda Basu, Marius Bozga, and Joseph Sifakis.
\newblock Modeling heterogeneous real-time components in {BIP}.
\newblock In {\em $4^{th}$ {IEEE} Int. Conf. on Software Engineering and Formal
  Methods ({SEFM06})}, pages 3--12, September 2006.
\newblock Invited talk.

\bibitem{dfinder}
Saddek Bensalem, Marius Bozga, Thanh-Hung Nguyen, and Joseph Sifakis.
\newblock {DFinder}: A tool for compositional deadlock detection and
  verification.
\newblock In {\em Computer Aided Verification}, pages 614--619. Springer, 2009.

\bibitem{esst4bip}
Simon Bliudze, Alessandro Cimatti, Mohamad Jaber, Sergio Mover, Marco Roveri,
  Wajeb Saab, and Qiang Wang.
\newblock Formal verification of infinite-state {BIP} models.
\newblock In {\em International Symposium on Automated Technology for
  Verification and Analysis}, pages 326--343. Springer, 2015.

\bibitem{MiSE14p25}
Simon Bliudze, Anastasia Mavridou, Radoslaw Szymanek, and Alina Zolotukhina.
\newblock Coordination of software components with {BIP}: Application to
  {OSGi}.
\newblock In {\em Proceedings of the 6th International Workshop on Modeling in
  Software Engineering}, MiSE 2014, pages 25--30, New York, NY, USA, 2014. ACM.

\bibitem{SPE:SPE2495}
Simon Bliudze, Anastasia Mavridou, Radoslaw Szymanek, and Alina Zolotukhina.
\newblock Exogenous coordination of concurrent software components with
  {JavaBIP}.
\newblock {\em Software: Practice and Experience}, page n/a, 2017.
\newblock Early view: \url{http://dx.doi.org/10.1002/spe.2495}.

\bibitem{alcon}
Simon Bliudze and Joseph Sifakis.
\newblock The algebra of connectors---structuring interaction in {BIP}.
\newblock {\em {IEEE} Transactions on Computers}, 57(10):1315--1330, 2008.

\bibitem{BliSif10-causal-fmsd}
Simon Bliudze and Joseph Sifakis.
\newblock Causal semantics for the algebra of connectors.
\newblock {\em Formal Methods in System Design}, 36(2):167--194, June 2010.

\bibitem{quilbeuf10-distr}
Borzoo Bonakdarpour, Marius Bozga, Mohamad Jaber, Jean Quilbeuf, and Joseph
  Sifakis.
\newblock From high-level component-based models to distributed
  implementations.
\newblock In {\em Proceedings of the tenth ACM international conference on
  Embedded software}, EMSOFT '10, pages 209--218, New York, NY, USA, 2010. ACM.

\bibitem{bozga2012modeling}
Marius Bozga, Mohamad Jaber, Nikolaos Maris, and Joseph Sifakis.
\newblock Modeling dynamic architectures using {Dy-BIP}.
\newblock In {\em Software Composition (SC 2012)}, volume 7306 of {\em LNCS},
  pages 1--16. Springer, 2012.

\bibitem{Bradbury:2004}
Jeremy~S. Bradbury, James~R. Cordy, Juergen Dingel, and Michel Wermelinger.
\newblock A survey of self-management in dynamic software architecture
  specifications.
\newblock In {\em Proceedings of the 1st ACM SIGSOFT Workshop on Self-managed
  Systems}, WOSS '04, pages 28--33, New York, NY, USA, 2004. ACM.

\bibitem{BRUNI200839}
Roberto Bruni, Antonio Bucchiarone, Stefania Gnesi, and Hern\'an Melgratti.
\newblock Modelling dynamic software architectures using typed graph grammars.
\newblock {\em Electronic Notes in Theoretical Computer Science}, 213(1):39 --
  53, 2008.

\bibitem{bruni14-bip-dynamics}
Roberto Bruni, Hern\'an~C. Melgratti, and Ugo Montanari.
\newblock Behaviour, interaction and dynamics.
\newblock In {\em Specification, Algebra, and Software - Essays Dedicated to
  Kokichi Futatsugi}, volume 8373 of {\em LNCS}, pages 382--401. Springer,
  2014.

\bibitem{canal1999specification}
Calos Canal, Ernesto Pimentel, and Jos{\'e}~M Troya.
\newblock Specification and refinement of dynamic software architectures.
\newblock In {\em Software Architecture}, pages 107--125. Springer, 1999.

\bibitem{Clarke08-ReconfLogic}
Dave Clarke.
\newblock A basic logic for reasoning about connector reconfiguration.
\newblock {\em Fundamenta Informaticae}, 82(4):361--390, 2008.

\bibitem{cuesta2001dynamic}
Carlos~E Cuesta, Pablo de~la Fuente, and Manuel Barrio-Sol{\'a}rzano.
\newblock Dynamic coordination architecture through the use of reflection.
\newblock In {\em Proceedings of the 2001 ACM symposium on Applied computing},
  pages 134--140. ACM, 2001.

\bibitem{Giusto2011}
Cinzia Di~Giusto and Jean-Bernard Stefani.
\newblock Revisiting glue expressiveness in component-based systems.
\newblock In {\em COORDINATION 2011}, pages 16--30. Springer, 2011.

\bibitem{edelmann17-FunctBIP}
Romain Edelmann, Simon Bliudze, and Joseph Sifakis.
\newblock Functional {BIP}: {Embedding} connectors in functional programming
  languages.
\newblock {\em Journal of Logical and Algebraic Methods in Programming}, 2017.
\newblock Under review.

\bibitem{georgiadis2002self}
Ioannis Georgiadis, Jeff Magee, and Jeff Kramer.
\newblock Self-organising software architectures for distributed systems.
\newblock In {\em Proceedings of the first workshop on Self-healing systems},
  pages 33--38. ACM, 2002.

\bibitem{hirsch1998graph}
Dan Hirsch, Paola Inverardi, and Ugo Montanari.
\newblock Graph grammars and constraint solving for software architecture
  styles.
\newblock In {\em Proceedings of the third international workshop on Software
  architecture}, pages 69--72. ACM, 1998.

\bibitem{inverardi1995formal}
Paola Inverardi and Alexander~L. Wolf.
\newblock Formal specification and analysis of software architectures using the
  chemical abstract machine model.
\newblock {\em Software Engineering, IEEE Transactions on}, 21(4):373--386,
  1995.

\bibitem{jackson2002alloy}
Daniel Jackson.
\newblock Alloy: a lightweight object modelling notation.
\newblock {\em ACM Transactions on Software Engineering and Methodology
  (TOSEM)}, 11(2):256--290, 2002.

\bibitem{Koehler08-dataflow}
Christian Koehler, David Costa, Jos{\'{e}} Proen{\c{c}}a, and Farhad Arbab.
\newblock Reconfiguration of {Reo} connectors triggered by dataflow.
\newblock {\em ECEASST}, 10, 2008.

\bibitem{Krause11-dyReo}
Christian Krause, Ziyan Maraikar, Alexander Lazovik, and Farhad Arbab.
\newblock Modeling dynamic reconfigurations in {Reo} using high-level
  replacement systems.
\newblock {\em Science of Computer Programming}, 76(1):23--36, 2011.

\bibitem{le1998describing}
Daniel Le~M{\'e}tayer.
\newblock Describing software architecture styles using graph grammars.
\newblock {\em Software Engineering, IEEE Transactions on}, 24(7):521--533,
  1998.

\bibitem{magee1996dynamic}
Jeff Magee and Jeff Kramer.
\newblock Dynamic structure in software architectures.
\newblock {\em ACM SIGSOFT Software Engineering Notes}, 21(6):3--14, 1996.

\bibitem{mavridou2016diagrams}
Anastasia Mavridou, Eduard Baranov, Simon Bliudze, and Joseph Sifakis.
\newblock Architecture diagrams: A graphical language for architecture style
  specification.
\newblock In {\em Proceedings of the 9th Interaction and Concurrency Experience
  (ICE)}, pages 83--97, August 2016.

\bibitem{mavridou2017configuration}
Anastasia Mavridou, Eduard Baranov, Simon Bliudze, and Joseph Sifakis.
\newblock Configuration logics: Modeling architecture styles.
\newblock {\em Journal of Logical and Algebraic Methods in Programming},
  86(1):2--29, 2017.

\bibitem{milner1992calculus}
Robin Milner, Joachim Parrow, and David Walker.
\newblock A calculus of mobile processes,~{I}.
\newblock {\em Information and computation}, 100(1):1--40, 1992.

\bibitem{Papadopoulos2001}
George~A. Papadopoulos and Farhad Arbab.
\newblock Configuration and dynamic reconfiguration of components using the
  coordination paradigm.
\newblock {\em Future Generation Computer Systems}, 17(8):1023--1038, 2001.

\end{thebibliography}
\end{document}